\documentclass[a4paper,UKenglish,cleveref, autoref, thm-restate]{lipics-v2021}

\newbool{lipics}
\boolfalse{lipics}

\hideLIPIcs

\title{Constructive higher sheaf models with applications to synthetic mathematics}
\subtitle{Extended version}
\titlerunning{Constructive higher sheaf models with applications to synthetic mathematics}

\author{Thierry Coquand}{University of Gothenburg, Sweden \and Chalmers University of Technology, Sweden}{coquand@chalmers.se}{}{}

\author{Jonas Höfer}{University of Gothenburg, Sweden \and Chalmers University of Technology, Sweden}{hoferj@chalmers.se}{https://orcid.org/0009-0001-9506-8475}{}

\author{Christian Sattler}{University of Gothenburg, Sweden \and Chalmers University of Technology, Sweden}{sattler@chalmers.se}{https://orcid.org/0000-0001-6374-4427}{US Air Force Office of Scientific Research, award number FA9550-24-1-0302}

\authorrunning{T. Coquand, J. Höfer, and C. Sattler} 

\Copyright{Thierry Coquand, Jonas Höfer, and Christian Sattler} 



\category{} 

\funding{ForCUTT project, ERC advanced grant 101053291}

\acknowledgements{%
We thank the anonymous referees for the helpful feedback on how to improve the presentation.
We thank Evan Cavallo for helpful discussions.%
}

\nolinenumbers 

\EventEditors{Claudia Faggian and Joost-Pieter Katoen}
\EventNoEds{2}
\EventLongTitle{41st Annual Symposium on Logic in Computer Science (LICS 2026)}
\EventShortTitle{LICS 2026}
\EventAcronym{LICS}
\EventYear{2026}
\EventDate{July 20--23, 2026}
\EventLocation{Lisbon, Portugal}
\EventLogo{}
\SeriesVolume{380}
\ArticleNo{58}

\newtheorem{assumption}{Assumption}
\newtheorem{construction}[theorem]{Construction}
\newtheorem{axiom}{Axiom}
\crefname{axiom}{Axiom}{Axioms}
\crefname{assumption}{Assumption}{Assumptions}

\usepackage{symbols}
\usepackage{macros}

\usepackage{tikz-cd}
\usepackage{todonotes}

\tikzcdset{arrow style=math font}
\tikzset{
  phantom/.style={/tikz/commutative diagrams/phantom},
  near/.style={phantom, very near start},
  pullback/.style={near, "\lrcorner"},
  pullback-mirrored/.style={near, "\llcorner"},
  pushout/.style={near, "\ulcorner"},
  pushout-mirrored/.style={near, "\urcorner"},
  weq/.style = {"\simeq"'{sloped,#1}},
  weq'/.style = {"\simeq"{sloped,#1}},
  equiv/.style = {"\cong"'{sloped,#1}},
  equiv'/.style = {"\cong"{sloped,#1}},
}

\begin{document}

\maketitle

\begin{abstract}
  There have recently been several developments in synthetic mathematics using extensions of dependent type theory with univalence and higher inductive types: simplicial homotopy type theory, synthetic algebraic geometry and synthetic Stone duality.
  We provide a foundation of higher sheaf models of type theory in a constructive metatheory and, in particular, build constructive models of these formal systems.
\end{abstract}

\section{Introduction}

In \emph{synthetic} (or \emph{axiomatic}) mathematics one uses a specialized formal language, such as an extension of homotopy type theory (\(\HoTT\))~\cite{HoTTBook2013}, as an internal language (\ie, domain-specific language) for a particular mathematical setting.
By extending the language with basic operations and axioms, one tries to axiomatically capture the core aspects of the chosen field.
This has been done, for example, for homotopy theory~\cite{HoTTBook2013,RiehlShulman2017,GratzerWeinbergerBuchholtz2025}, domain theory~\cite{SterlingYe2025}, probability theory~\cite{Simpson2024}, algebraic geometry~\cite{CherubiniCoquandHutzler2024}, and light condensed sets~\cite{CherubiniCoquandGeerligsMoeneclaey2024}.
The synthetic approach often allows for more succinct and conceptual definitions and proofs compared to the traditional \emph{analytic} point of view, since technical preliminaries are abstracted away.
This can make the mathematics more accessible and also better-suited for formalization.
For example, in synthetic homotopy theory done in \(\HoTT\), one can define complicated objects such as Eilenberg-MacLane spaces or cohomology groups directly~\cite{HoTTBook2013,CherubiniCoquandHutzler2024}.

The work done in a synthetic setting is connected to the traditional analytic results by means of a model construction.
One provides an interpretation of the language into suitable mathematical objects in such a way that the axiomatized synthetic objects are interpreted as their analytic counterparts.
In a sense, the simplicity of working in the synthetic setting is bought by dealing---once and for all---with certain technicalities in the model construction.

A standard source of models is given by categories of sheaves. 
As a model of type theory, a category of sheaves can be thought of as a generalized category of sets: it has enough structure to perform most basic mathematical constructions.
In the \(1\)-categorical setting, sheaf models are extremely flexible and interpret a rich fragment of extensional type theory~\cite{Hofmann1997}.
Furthermore, for many applications, key objects naturally organize as or embed into a category of sheaves.
\emph{Higher} sheaf models are the \(\infty\)-categorical generalization.
They instead generalize the higher category of spaces (which includes all sets, but also higher objects such as spheres) and instead interpret a form of \(\HoTT\).
Even if one is a priori just interested in set-level objects, the generalization to higher sheaves is still useful because it gives access to tools provided by \(\HoTT\), such as the synthetic definition of cohomology groups mentioned above. 

An interpretation of \(\HoTT\) in higher sheaf categories in a classical metatheory was given by Shulman~\cite{Shulman2019} using type-theoretic model topoi.
We instead work in the constructive framework of \emph{cubical models}~\cite{Coquand2018Survey,OrtonPitts2018}.
More generally, we work with certain \emph{inner models} within these, refinements where types carry additional structure.
This class of models contains important examples, including a forthcoming model due to Sattler~\cite{Sattler2025Types} that constructively presents the higher category of spaces.
The first goal of this paper is to extend this framework with a construction of presheaves on an internal category in a given \emph{base} model.
By exploiting the framework of lex modalities~\cite{RijkeShulmanSpitters2020,CoquandRuchSattler2021}, this allows for the construction of higher sheaf models of \(\HoTT\), supporting higher inductive types (by applying the results of~\cite{CoquandHuberMortberg2018,CoquandRuchSattler2021}).
We conjecture that these models present higher categories of sheaves.
The second goal is to show that the resulting models validate axioms used in synthetic developments.
In particular, we show that Blechschmidt's duality axiom~\cite{Blechschmidt2019}, central to many synthetic developments, holds in suitable higher presheaf categories.

\subsection{Contributions}

All our contributions happen in a constructive and, in particular, predicative metatheory.

\begin{itemize}
  \item We construct models of \(\HoTT\) in internal presheaves over an arbitrary cubical category.
    These allow for the construction of sheaf models using lex modalities, and support universes closed under higher inductive types by previous work~\cite{CoquandHuberMortberg2018,CoquandRuchSattler2021}.
  \item We give a constructive justification of synthetic algebraic geometry~\cite{CherubiniCoquandHutzler2024,projective2025}.
    This fixes a subtle mistake in the original model construction (cf.~\cite{CherubiniCoquandHutzler2024,CoquandHoferSattler2025}).
  \item More generally, we show that Blechschmidt's duality axiom~\cite{Blechschmidt2019} holds in our cubical presheaf model constructed on any category of models in h-sets of an algebraic theory.
  \item The base model of our presheaf construction can be taken to be a suitable inner model.
    The presheaf model inherits expected principles such as dependent choice.
    If the base model presents spaces~\cite{Sattler2025Types}, we expect the presheaf model to present higher presheaves.
\end{itemize}

\subsection{Internal/external yoga}

A distinctive feature of our approach is its methodology.
The statements we ultimately obtain are \emph{external} semantic results: they describe models of \(\HoTT\) and modalities on categories of cubical presheaves.
However, since presheaves are a model of extensional type theory, many proofs are carried out \emph{internally}.
For example, external results about cubical presheaves often reduce to results about ordinary presheaves internally to cubical sets.
This internal/external correspondence functions as a ``yoga'': proofs that would be long when written purely externally become short and conceptual when expressed internally.
This technique has, for example, been used by Orton and Pitts~\cite{OrtonPitts2018} to construct models of cubical type theory.

\subsection{Outline}

In \cref{sec:preliminaries}, we recall basic constructions on models of type theory used in the rest of the paper: presheaf models and inner models.
As an example of inner models, we include models of \(\HoTT\) via cubical models.
In \cref{sec:setup}, we  establish some running assumptions and construct a first na\"{\i}ve candidate for a higher presheaf model.
In \cref{sec:cobar}, we generalize the construction of the cobar modality by Coquand, Ruch, and Sattler~\cite{CoquandRuchSattler2021} from cubical presheaves over a category to presheaves over a category internal to a cubical model.
This modality is the key tool for refining the na\"{\i}ve presheaf models.
In the inner model of modal types, equivalences are characterized as desired as levelwise equivalences in the base model.
In \cref{sec:lifting-of-descent-data-operations}, using a technical assumption on the base category, we relate the constructed presheaf model to the base model.
Furthermore, we extend the construction to suitable inner models of a cubical model.
We then provide a number of results that lift properties from the base to the presheaf model.
In \cref{sec:duality}, we provide applications of our work.
We show that our higher presheaf models, when constructed on the category of models in h-sets of an algebraic theory, validate Blechschmidt's duality axiom~\cite{Blechschmidt2019}.
In particular, we construct a model of the axioms of synthetic algebraic geometry.
In \cref{sec:representables}, we define the notion of locally homotopy representable type and characterize dependent products over them.
This is used in \cref{sec:duality}.

\section{Preliminaries}\label[section]{sec:preliminaries}

The metatheory will be constructive set theory, with a cumulative hierarchy of universes \(\ExtUniv_n\), as considered by Aczel~\cite{aczel:relate}.
Alternatively, it can be extensional type theory, also with a hierarchy of universes (and indeed this will be used when reasoning internally).
Our development is entirely algebraic, in particular we do not need quotients in the metatheory.

In this text, equality denoted by the symbol $=$ refers to strict equality.
We write $\cong$ for strict isomorphisms.
When working in \(\HoTT\), the identity or path type is denoted by the symbol $\simeq$ and $f\sim g$ is used for homotopies $\prod_{x\co A}f(x)\simeq g(x)$ between functions.

\subsection{Models of type theory}

We will consider three kinds of type theory: bare type theory, extensional type theory~\cite{ML84}, and homotopy type theory (or univalent type theory)~\cite{HoTTBook2013}.
The notion of model we consider is similar to that of model of an equational theory, with operations that have a fixed arity and which should satisfy some equations.

The definitions are standard~\cite{Dybjer95,Ehrhard1988Thesis,Hofmann1997}, but we recall the main points to fix the notations.
A \emph{model} $M$ of bare type theory is given by collections $\Cont$, $\Subs(\Delta,\Gamma)$, $\Ty(\Gamma)$, $\El(\Gamma,A)$ with operations of fixed arity satisfying some equations.
The sort $\Cont$ is the sort of \emph{contexts}, written $\Gamma,\Delta,\dots$.
Given contexts $\Delta$ and $\Gamma$, we have a collection $\Subs(\Delta,\Gamma)$ of \emph{substitutions} that we write $\Delta \rightarrow \Gamma$.
We have identity substitutions \(\id \co \Gamma \rightarrow \Gamma\) and composite substitutions \(\delta \circ \sigma\) (also written \(\delta \sigma\)) for $\delta$ in $\Subs(\Delta,\Gamma)$ and $\sigma$ in $\Subs(\Theta,\Delta)$ satisfying neutrality and associativity equations.%
\footnote{%
We stress that, like for equational theories, each operation has a fixed arity.
That is, the complete notation is $\id(\Gamma) \in \Subs(\Gamma,\Gamma)$ for $\Gamma\in \Cont$ and $\comp(\Theta,\Delta,\Gamma,\delta,\sigma)\in \Subs(\Theta,\Gamma)$ for $\Gamma,\Delta,\Theta$ in $\Cont$ and $\sigma$ in $\Subs(\Delta,\Gamma)$ and $\delta$ in $\Subs(\Theta,\Delta)$.
We write $\id$ and $\delta \sigma$ as informal notations, following~\cite{CARTMELL86}.%
}
We have a context $1$ such that $\Gamma \rightarrow 1$ is a singleton; we write $\pair{}$ for its unique element.
For a context $\Gamma$, we have a collection $\Ty(\Gamma)$ of \emph{types} in the context $\Gamma$.
Given further $A\in\Ty(\Gamma)$, we have a collection $\El(\Gamma,A)$ of \emph{elements} of type $A$.
Given $\sigma\co\Delta\rightarrow\Gamma$, we have substitution operations $A\sigma\in\Ty(\Delta)$ for $A \in \Ty(\Gamma)$ and $t\sigma\in\El(\Delta,A\sigma)$ for $t\in\El(\Gamma,A)$.%
\footnote{We follow Martin-L\"of's notations for explicit substitution.}
These are compatible with identity and composition.

A \emph{comprehension structure} assigns to each $A$ in $\Ty(\Gamma)$ a context extension $\Gamma.A$ together with a projection $\p\co\Gamma.A\rightarrow \Gamma$ and a generic element $\q\in\El(\Gamma.A,A\p)$ such that, for every $\sigma\co\Delta\rightarrow\Gamma$ and $u$ in $\El(\Delta,A\sigma)$, there is a unique map $\pair{\sigma,u}\co\Delta\rightarrow\Gamma.A$ such that $\p\pair{\sigma,u} = \sigma$ and $\q\pair{\sigma,u} = u$.
Uniqueness is captured by the further equations $\pair{\sigma,u}\delta = \pair{\sigma\delta,u\delta}$ and $\pair{\p,\q} = \id$.%
\footnote{Note that this implies $\vartheta = \pair{\p\vartheta,\q\vartheta}$ for $\vartheta\co\Delta\rightarrow\Gamma.A$.}
Given $\sigma\co\Delta\rightarrow\Gamma$ and $A$ in $\Ty(\Gamma)$, we use the notation $\sigma^+$ for the lifting $\pair{\sigma\p,\q}\co\Delta.A\sigma\rightarrow\Gamma.A$.

To build a model of bare type theory is to define collections corresponding to the sorts and operations satisfying the given equations.
We may extend bare type theory with a stratification of $\Ty$, introducing a cumulative stratification $\Ty_n(\Gamma)$ of $\Ty(\Gamma)$.

To get models of extensional~\cite{ML84} or homotopy type theory~\cite{HoTTBook2013},
we extend this notion of model by adding further operations and equations, dependent products and sums, path/identity types, and universes.
For universes, given a stratification $\Ty_n$, we require each of them to be closed under these operations, and we add $\Univ_n\in\Ty_l(\Gamma)$ for $n<l$ with an isomorphism $\El(\Gamma,\Univ_n)\cong \Ty_n(\Gamma)$ natural in $\Gamma$.
For extensional type theory, we require the identity type to satisfy the equality reflection rule~\cite{ML84},
and for homotopy type theory, we require each universe to satisfy the univalence axiom~\cite{HoTTBook2013}.
We can also extend the models with higher inductive types~\cite{CoquandHuberMortberg2018,HoTTBook2013}.

\subsection{Building models of type theory}

We review two main methods of building models of type: presheaf models and inner models.
\emph{All models we consider in this article are inner models of suitable presheaf models.}

\subsubsection{Presheaf models of extensional type theory}

Let $\cat*{C}$ be a category in $\ExtUniv_0$.
We define the \emph{presheaf model} of extensional type theory on $\cat*{C}$, following~\cite{Hofmann1997}.%
\footnote{This has been formalized in NuPrl in~\cite{bickford20}.}
We write $X,Y,Z,\dots$ for objects of $\cat*{C}$ and $f,g,\dots$ for arrows of $\cat*{C}$.

A context is a presheaf on $\cat*{C}$ with values in some $\ExtUniv_m$, \ie, a collection of elements $\Gamma(X)$ of $\ExtUniv_m$, for some $n$, with restriction operations $\rho\mapsto \rho f \co \Gamma(X)\rightarrow\Gamma(Y)$ for $f\co Y\rightarrow X$ compatible with composition.
A substitution $\sigma\co\Delta\rightarrow \Gamma$ is a natural transformation.
A type $A$ in $\Ty(\Gamma)$ is a family of sets $A(X,\rho)$ in some $\ExtUniv_n$ with restriction operations $u\mapsto u f \co A(X,\rho)\rightarrow A(Y,\rho f)$ compatible with composition.
An element $t$ in $\El(\Gamma,A)$ is a section: a family of elements $t(X,\rho)$ in $A(X,\rho)$ such that $t(X,\rho)f = t(Y,\rho f)$.
The context extension $\Gamma.A$ is defined by taking $(\Gamma.A)(X)$ to be the set of pairs $(\rho,u)$ with $\rho$ in $\Gamma(X)$ and $u$ in $A(X,\rho)$.
We write $\yo X$ for the presheaf represented by $X$.
We have a stratification $\Ty_n(\Gamma)$ of $\Ty(\Gamma)$ by imposing $A(X,\rho)$ to be in $\ExtUniv_n$.
Each universe $\ExtUniv$ of the metatheory is reflected by a universe in the presheaf model.
We define $\Univ_n(X,\rho) = \Univ_n(X) = \Ty_n(\yo X)$.

\subsubsection{Inner models}

Inner models are an important way to obtain a new model (similar to the one used for constructible sets in set theory) from a given model $M$.
See \eg~\cite{Ruch2022}.

An \emph{internalisation structure} on $M$ is (stratification-preserving) operation $C(A)\in\Ty(\Gamma)$ for $A\in\Ty(\Gamma)$ satisfying $C(A)\sigma = C(A\sigma)$ on the types of $M$.
We further assume operations
\[\begin{gathered}
  \pi_C(a, b)\in\El(\Gamma,C(\Pi_AB)),
  \qquad
  \sigma_C(a, b)\in\El(\Gamma,C(\Sigma_AB)),
  \qquad
  u^i_C\in\El(\Gamma,C(\Sigma_{\Univ_i}C(\q)))
\end{gathered}\]
where $a\in\El(\Gamma,C(A))$ and $b\in\El(\Gamma.A,C(B))$ natural in $\Gamma$.
We can then define an \emph{inner} model $M_C$ by taking the contexts and substitutions of $M$ and defining $\Ty^C(\Gamma)$ to be the collection of pairs $(A,a)$ with $A\in\Ty(\Gamma)$ and $a\in\El(\Gamma,C(A))$ and defining $\El^C(\Gamma,(A,a))$ to be $\El(\Gamma,A)$.
It is remarkable that we need only constants for reflecting operations, without further equational conditions.

To get a model of \(\HoTT\), one examines each remaining type former and either shows that \(C\) is closed under it or justifies it directly.
Important examples of inner models are the following and \emph{cubical models}, which we cover in the next section.

\begin{example}
  Starting from a model of \(\HoTT\), we may obtain an inner model by \emph{nullification} of a family of propositions $P_i$.
  For this, we take \(C(A)\) to be the h-proposition expressing that all constant maps \(A \rightarrow A^{P_i}\) are equivalences.
  In this case, \(C\) is also closed under identity types, while (higher) inductive types such as coproducts, truncations, and the natural numbers are obtained (modulo computation rules) by reflection to \(P_i\)-null types~\cite{RijkeShulmanSpitters2020}.%
\end{example}

\subsubsection{Cubical models}

For presheaf categories with sufficient structure, we can construct a inner model of \(\HoTT\) of the model of extensional type theory.
These are similar to the simplicial sets model~\cite{KapulkinLumsdaine2021}.%
\footnote{See~\cite[Appendix~D]{CoquandHuberSattler2022} for how to obtain a structured version of the simplicial set model using the machinery of inner models assuming excluded middle, but not choice.}

\begin{definition}
  A \emph{cubical model setup} is a category \(\Box\) in \(\ExtUniv_0\) equipped with an \emph{interval} \(\I \in \PSh{\Box}\) and a \emph{cofibration classifier} \(\Cof \subseteq \Omega_{\mathsf{dec}}\) satisfying the axioms from~\cite{Coquand2018Survey} or \cite{LicataOrtonPittsSpitters2018,OrtonPitts2018}.
\end{definition}

For \(A\) in \(\Ty(\Gamma)\), we define a set \(\Fib(\Gamma,A)\) of open box filling or fibrancy \emph{structures}~\cite{CohenCoquandHuberMortberg2015,OrtonPitts2018}.
This is in contrast to the simplicial model where horn filling is a \emph{property}.
We can then construct\footnote{This is not provable for simplicial sets if \(\Fib(\Gamma,A)\) is the horn filling property~\cite{bezem15}.} a map \(\Fib(\Gamma,A)\times\Fib(\Gamma.A,B) \to \Fib(\Gamma,\Pi_AB)\) natural in $\Gamma$, and similarly for dependent sums and universes as required for an inner model.
Identity types are justified from paths defined using the interval \(\I\).

To prove these closure properties, it is convenient~\cite{OrtonPitts2018} to use that \(\PSh{\Box}\) forms a model of extensional type theory.
Internally to \(\PSh{\Box}\), for \(\Gamma \co \Univ_n\) and \(A \co \Gamma \to \Univ_n\) one defines a type \(\Fib*(\Gamma, A)\) such that globally, sections of \(\Fib*(\Gamma, A)\) are given by \(\Fib(\Gamma,A)\).
Using the language of dependent type theory, we can then construct \emph{internally} a family of maps
\[
  \prod_{\Gamma\co\Univ_n} \prod_{A\co\Gamma\to \Univ_n} \prod_{B\co\Gamma.A \to \Univ_n} \Fib*(\Gamma,A)\times \Fib*(\Gamma.A,B) \longrightarrow \Fib*(\Gamma,\Pi_AB)
\]
natural in \(\Gamma\).
Here, we use the internal reasoning as intermediate steps to prove a global statement.
It is possible~\cite{CohenCoquandHuberMortberg2015,CoquandHuberSattler2022,LicataOrtonPittsSpitters2018} to define \(C(A)\in\Ty(\Gamma)\) for \(A\in\Ty(\Gamma)\) with an isomorphism \(\El(\Gamma,C(A))\cong \Fib(\Gamma,A)\).
This however cannot be carried out internally to extensional type theory~\cite{LicataOrtonPittsSpitters2018}.

The corresponding inner model is a model of homotopy type theory.
We write \(\Ty^{\Fib}\) for \(\Ty^C_{\PSh{\Box}}\), indicating that we are working with types from this inner model.

\subsection{Pseudomorphisms, lex operations, and descent data operations}

\newcommand\Inv{{\downarrow}}

An important notion is the notion of pseudomorphism \(D \co M \rightarrow N\) of models of bare type theory~\cite{kaposi_et_al:LIPIcs.FSCD.2019.25,Ruch2022,CoquandRuchSattler2021}.
This acts on contexts, substitutions, types (\(D A\) in \(\Ty_N(D \Gamma)\) for \(A\) in \(\Ty_M(\Gamma)\)), and elements.
To be a pseudomorphism, $D$ is required to preserve composition and substitution operations, but $D$ is not required to commute \emph{strictly} with the extension operation $\Gamma.A$; we just ask that the comparison maps \(\pair{D \p,D \q} \co D (\Gamma.A)\rightarrow D \Gamma.D A\) and \(D 1\rightarrow 1\) have an \emph{inverse}.%
\footnote{So the only new operations we add are $\Inv \co D \Gamma.D  A\to D (\Gamma.A)$ and $\Inv\co1\rightarrow D 1$ with the equations $\Inv\pair{D \p,D \q} = \id$ and $\pair{D \p,D \q}\Inv = \id$, which is equivalent to the pair of equations $(D \p)\Inv = \p$ and $(D \q)\Inv = \q$, and $\Inv\pair{} = \id$.}

A \emph{left exact} morphism of models of homotopy type theory is a pseudomorphism between the corresponding models of bare type theory which furthermore sends contractible types to contractible types.
If we see a model of homotopy type theory as a notion of predicative elementary higher topos, this can be seen as a higher analog of the notion of left exact morphisms between toposes introduced in~\cite{SGA4}.

A \emph{pointing} \(\alpha \co \Gamma\rightarrow D \Gamma\) of a pseudo-endomorphism \(D\) on a model \(M\) consists of maps \(\alpha_\Gamma \co\Gamma\rightarrow D \Gamma\) for each \(\Gamma\) satisfying \(\alpha \sigma = (D \sigma) \alpha\) for $\sigma \co \Delta \to \Gamma$.

\begin{theorem}[\cite{CoquandRuchSattler2021,Ruch2022}]
  Pointed pseudo-endomorphisms on models of \(\HoTT\) are left exact.
\end{theorem}

Given a pointed pseudo-endomorphism as above, we obtain a \emph{lex operation} $\LO{D}A$ on types by setting $\LO{D}A \coloneqq (D  A)\alpha$, so that $\LO{D}A\in\Ty(\Gamma)$ for $A\in\Ty(\Gamma)$ and $(\LO{D}A)\sigma = \LO{D}(A\sigma)$.%
\footnote{As in~\cite{CoquandRuchSattler2021,Ruch2022}, we also have an action on families $\widetilde{\LO{D}}B \coloneqq (DB)\Inv\alpha^+$ in $\Ty(\Gamma.\LO{D}A)$ for $B$ in $\Ty(\Gamma.A)$ (which satisfies $(\widetilde{\LO{D}}B)(\eta_A a) = \LO{D}(B(a))$) and on elements $\widetilde{\LO{D}}b \coloneqq (Db)\Inv\alpha^+$ in $\El(\Gamma.\LO{D}A,\widetilde{\LO{D}}B)$ for $b$ in $\El(\Gamma.A,B)$.}
The operation $\LO{D}$ extends to a functor on types, with a unit map $\eta_A\co A\rightarrow \LO{D}A$ for $A$ in $\Ty(\Gamma)$.

\begin{definition}
  A lex operation $\LO{D} $ is a \emph{descent data operation} iff both $\LO{D}\eta_A$ and $\eta_{\LO{D}A}$ are equivalences for all $A$.
\end{definition}

In this case, $\LO{D}$ is now a left exact modality (\cf~\cite[\S2.3]{CoquandRuchSattler2021}) in the sense of~\cite{RijkeShulmanSpitters2020}, which satisfies additional strictness conditions~\cite{CoquandRuchSattler2021}.
We get then an inner model by taking $C(A)$ to be the h-proposition stating that $A$ is $\LO{D}$-modal, \ie, that $\eta_A$ is an equivalence.
We write $\Ty^{\sD}$ to indicate that we are working with types from this inner model.
These strictness conditions are important for justifying (higher) inductive types in concrete situations~\cite{CoquandRuchSattler2021}.%
\footnote{For instance, the type of natural numbers is now the (higher) inductive type $\sN$ with, besides the usual constructors $\zero\co\sN$ and $\succ\co\sN\to\sN$, new higher constructors $\patch\co\LO{D}\sN\to\sN$ and $\linv\co\Pi_{x\co\sN}\patch (\eta x)\simeq x$ forcing $\sN$ to be $\LO{D}$-modal.}

In the situation where $M$ is an inner cubical mode, $D$-modal types are inherently fibrant.
In particular, we have a forgetful natural transformation $\Ty^{\sD} \to \Ty^{\Fib}$.

Examples of descent data operations are the cobar modality (which we study in \Cref{sec:cobar}), nullification of a proposition, and Sattler's modality~\cite{Sattler2025Types}.

\section{Setup}\label[section]{sec:setup}

In this section, we establish some running assumptions and construct the na\"{\i}ve candidate for the higher presheaf model on a cubical category.

\subsection{The external setup}\label[subsection]{sec:setup:external-setup}

We fix our \emph{base} cubical model of homotopy type theory.
In \cref{sec:lifting-of-descent-data-operations}, we will refine it to an inner model induced by a descent data operation.
The following assumption (and the refinements \cref{ass:base-model-lex-operation,ass:base-model-descent-data-operation}) are standing assumptions for the rest of the paper.

\begin{assumption}\label[assumption]{ass:cubical-setup}
  Let \((\Box, \I, \Phi)\) be a cubical model setup such that \(\I\) has connections and the map \(\I \to 1\) is locally representable.
\end{assumption}

We write \(\cSet\) for \(\PSh{\Box}\).
The map \(\I \to 1\) being locally representable means that \(\yo x \times \I\) is representable for all \(x \in \Box\), \ie, that \(\I \times -\) preserves representables.
This implies that \(\I\) is tiny.
Hence, \(\cSet_{\Fib}\) is a model of \(\HoTT\)~\cite{OrtonPitts2018,LicataOrtonPittsSpitters2018} supporting higher inductive types~\cite{CoquandHuberMortberg2018}.

We fix a cubical category which serves as the site for our presheaf model construction.
This assumption (and the later refinement \cref{ass:fibrant-cat}) are standing until the applications of the construction in \cref{sec:duality}, where we instantiate them suitably for models of duality.

\begin{assumption}\label[assumption]{ass:cat}
  Let \(\cat{C}\) be a category in \(\Univ_0\) internal to \(\cSet\).
\end{assumption}

A category \(\cat{C}\) internal to \(\CubicalSet\) is the same as a category defined in the internal language of \(\cSet\) in the empty context.
This corresponds to a functor \(\cat{C}^\op \to \Cat\), \ie, for \(I \in \Box\), a category \(\cat{C}(I)\) with, for \(f \co J \to I\), a functor \(\cat{C}(I) \to \cat{C}(J)\), respecting composition in \(I\).
The \emph{Grothendieck construction}~\cite{SGA1} of \(\cat{C}\) is the category \(\cat*{C}\) in \(\ExtUniv_0\) whose objects are given by pairs \((I,x)\) with \(x\) in \(\cat{C}_0(I)\) and for which a morphism \((J, y) \to (I, x)\) is given by a pair of morphisms \(f \co J \to I\) and \(\alpha \co \cat{C}_1(J)(y, xf)\).

\Cref{ass:cubical-setup} on \(\Box\) descends to \(\cat*{C}\).
The interval and cofibration classifier are given by restriction along the projection from \(\cat*{C}\) to \(\Box\).
Hence, we can form a new model \(\PSh{\cat*{C}}_{\Fib}\) of \(\HoTT\) with higher inductive types as an inner model of the presheaf model \(\PSh{\cat*{C}}\).

\subsection{The internal setup}

We work internally to \(\CubicalSet\).
In the internal language, \(\cat{C}\) is given by a type \(\cat{C}_0\) of objects, a family \(\cat{C}_1\) over \(\cat{C}_0 \times \cat{C}_0\) of morphisms, and the usual operations for identity morphisms and composition, satisfying the usual (strict) equations.
We can form the presheaf model $\PSh*{\cat{C}}$ of extensional type theory.
Furthermore, we can lift the interval \(\I\) and cofibration classifier \(\Phi\) as constant presheaves to \(\PSh*{\cat{C}}\).
Hence, we have that basic notions such as fibrancy structures, paths, and equivalences are also defined for \(\PSh*{\cat{C}}\).

If $J$ is a type and we have a family $A_j \co \Ty*_{\PSh*{\cat{C}}}(\Gamma)$ for $j \co J$, then we can form the product $\prod_{j\co J} A_j$ levelwise.
In particular, path types are defined as a subtype of products over the interval.
Note that elements of path types correspond to paths in the type of elements, and that homotopies correspond to paths in the hom-types.

A fibrancy structure \(\alpha \co \Fib*(\Gamma,A)\) on a family \(A\) over a type \(\Gamma\) is an operation that given \(e \co \set{0,1}\), \(\varphi \co \Phi\), \(\gamma \co \I \to \Gamma\), and \(a \co \prod_{i \co \I} [i=e\vee\varphi] \to A\gamma_i\) yields \(\alpha(e, \varphi, \gamma, a, i)\)  of type \(A\gamma_i[\varphi \vee i=e \mapsto a_i]\) for all \(i \co \I\).
Similarly, a trivial fibrancy structure on such a family is an operation that for every \(\gamma \co \Gamma\) completes a partial element \(u \co [\varphi] \to A\gamma\) to a total element \(A\gamma[\varphi \mapsto u]\).
This notion coincides with fiberwise contractibility on fibrant families.
A homogeneous fibrancy structure on such a family is a fibrancy structure on the family \(A\gamma\) over \(1\) for all \(\gamma \co \Gamma\).

For a presheaf \(\Gamma\) over \(\cat{C}\) and \(A\) type over \(\Gamma\), we can describe \(\Fib*_{\PSh*{\cat{C}_0}}(\Gamma, A)\), the type of global sections of the type of fibrancy structures, similar to~\cite[Section 3.3]{CoquandRuchSattler2021}.
An element of \(\Fib*_{\PSh*{\cat{C}_0}}(\Gamma, A)\) corresponds to a family of fibrancy structures \(\Fib*(\Gamma_x, A_x)\) for \(x \co \cat{C}_0\), such that the restriction operation for \(A\) preserves the operation given by the fibrancy structure.

We also have the presheaf model \(\PSh*{\cat{C}_0}\) where contexts are simply families over \(\cat{C}_0\).
There is a forgetful pseudomorphism \(U \co \PSh*{\cat{C}} \rightarrow \PSh*{\cat{C}_0}\).
We say a type \(A\) is \emph{levelwise fibrant} if \(UA\) is fibrant in the sense of \(\PSh*{\cat{C}_0}\).
This means it has the same structure as above, but restriction does not satisfy the compatibility condition.
We write \(\Ty*^{\Fib*_0}_{\PSh*{\cat{C}}}\) for the presheaf of levelwise fibrant types.
Similarly, we have levelwise notions for trivial fibrancy structures and equivalences.
We record the following observation about \(U\) preserving products.
Note that this mean externally that \(U\) preserves products indexed by cubical sets.

\begin{lemma}\label[lemma]{prop:forgetful:preservation-properties}
  The pseudomorphism \(U\) preserves products and path types, and is therefore in particular left exact.
\end{lemma}

These internal models correspond in the expected way to the external models, when externalized in the \emph{empty} context of \(\cSet\).
Global sections of \(\PSh*{\cat{C}}\), \(\Ty*_{\PSh*{\cat{C}}}(\Gamma)\), and \(\El*_{\PSh*{\cat{C}}}(\Gamma, A)\) correspond to elements of \(\PSh{\cat*{C}}\), \(\Ty_{\PSh{\cat*{C}}}(\Gamma)\), and \(\El_{\PSh{\cat*{C}}}(\Gamma, A)\).
The internal lifting of \(\I\) and \(\Phi\) corresponds externally to the construction from \cref{sec:setup:external-setup}.
In the empty context of \(\cSet\), types of \(\PSh*{\cat{C}}_{\Fib*}\) correspond to externally to types of \(\PSh{\cat*{C}}_{\Fib}\).

\section{The cobar modality}\label[section]{sec:cobar}

The goal of this section is to define a descent data operation \(\sD\) on the model \(\PSh{\cat*{C}}_{\Fib}\) where \(\cat{C}\) is a category internally to \(\cSet\).
This descent data operation has the key property that for \(\sD\)-modal types levelwise equivalences coincide with equivalences.
The inner model \(\PSh{\cat*{C}}_{\sD}\) of \(\sD\)-modal types is in this sense the correct model.
This operation will be crucial in \cref{sec:lifting-of-descent-data-operations} to lift a descent data operation from \(\cSet_{\Fib}\) to \(\PSh{\cat*{C}}_{\Fib}\).
This generalizes and refines some results by Coquand, Ruch, and Sattler~\cite{CoquandRuchSattler2021}.

We apply the internal-external yoga by defining the operation and proving its key properties internally.
In this entire section, we work internally to \(\cSet\).

\subsection{Cofree presheaves}

The descent data operation \(\LO{D}\) is obtained as a semisimplicial version of the cobar construction of another lex operation \(\LO{E}\).
This operation comes from the pseudomorphism that cofreely equips a family of types with a functorial action.
The pseudomorphism \(\PM{U} \co \PSh*{\cat{C}} \to \PSh*{\cat{C}_0}\) is given by restriction along the inclusion of objects \(\cat{C}_0 \to \cat{C}\).
As such, it always has a right adjoint pseudomorphism.

\begin{definition}
  The pseudomorphism \(\PM{R} \co \PSh*{\cat{C}_0} \to \PSh*{\cat{C}}\) is the right adjoint to \(\PM{U}\), given on contexts \(\Gamma \co \PSh*{\cat{C}_0}\), types \(A \co \Ty*_{\PSh*{\cat{C}_0}}(\Gamma)\), and elements \(a \co \El*_{\PSh*{\cat{C}_0}}(\Gamma, A)\) by
  \[
    (\PM{R} \Gamma)_x = \prod_{\substack{y \co \cat{C}_0 \\ f \co \cat{C}_1(y, x)}} \Gamma_y,
    \qquad
    (\PM{R} A)\gamma = \prod_{\substack{y \co \cat{C}_0 \\ f \co \cat{C}_1(y, x)}} A\gamma f,
    \qquad
    (\PM{R} a)\gamma_f = a(\gamma f).
  \]
  The unit \(\eta_\Gamma \co \Gamma \to \PM{R}\PM{U}\Gamma\) and counit \(\varepsilon_\Gamma \co \PM{U}\PM{R}\Gamma \to \Gamma\) are given by \(\paren{\eta_\Gamma\gamma}_f = \gamma f\), \(\paren{\varepsilon_\Gamma\gamma}_x = \gamma_{\id_x}\).
\end{definition}

\begin{remark}\label[remark]{prop:cofree-presheaf:dependent-right-adjoint}
  As a right adjoint pseudomorphism, we have an induced right adjoint on types.
  For \(\Gamma \co \PSh*{\cat{C}}\) and \(A \co \Ty*_{\PSh*{\cat{C}_0}}(\PM{U}\Gamma)\) we have \(\LO{\sR} A \coloneqq (RA)\eta_\Gamma \co \Ty*_{\PSh*{\cat{C}}}(\Gamma)\).
  This construction is natural in \(\Gamma\) and equipped with a natural isomorphism
  \(
    \El*_{\PSh*{\cat{C}_0}}(\PM{U}\Gamma, A)
    \cong \El*_{\PSh*{\cat{C}}}(\Gamma, \LO{R} A)
  \).
  Furthermore, \(\LO{R}\) extends to a functor, acting on functions between types.
  Since homotopies are paths in hom-types, the bijection \(\Ty*[\PSh*{\cat{C}_0}](\PM{U}\Gamma)(\PM{U}A, B) \cong \Ty*[\PSh*{\cat{C}}](\Gamma)(A, \LO{R}B)\) extends to homotopies.
  Compared to an ordinary right adjoint on types, we have more structure.
  For a family \(B \co \Ty*_{\PSh*{\cat{C}_0}}(\PM{U}\Gamma.A)\) we have another family \(\LO{R}_\Gamma B \coloneqq (RB)\eta_\Gamma^+ \co \Ty*_{\PSh*{\cat{C}}}(\Gamma.\LO{R} A)\).
  This operation preserves \(\Sigma\)-types and constant families in the sense that \(\LO{R}(\Sigma_A B) \cong \Sigma_{\LO{R} A} \paren{ \LO*{R}_{\Gamma} B }\) and \(\LO*{R} (B\p_A) = \LO{R}(B)\p_{\LO{R} A}\).
\end{remark}

Similar to \(U\), the pseudomorphism \(R\) also preserves products.

\begin{lemma}\label[lemma]{prop:cofree-presheaf:preservation-properties}
  The pseudomorphism \(\PM{R}\) preserves products and path types, and is therefore in particular left exact.
\end{lemma}
\begin{proof}
  Preservation of products follows directly from the definition since products are calculated levelwise and products commute with products.
  Preservation of paths follows since pseudomorphisms preserve \(\Sigma\)-types and extensional identity types.
\end{proof}
We can adapt a result in~\cite[Proposition~14]{CoquandRuchSattler2021} to show that \(\PM{R}\) lifts to fibrant types.
By \cref{prop:cofree-presheaf:preservation-properties}, the claim reduces to showing the case of trivial fibrancy.

\begin{lemma}\label[lemma]{prop:cofree-presheaf:fibrancy}
  Naturally in \(\Gamma \co \PSh*{\cat{C}_0}\), given \(A \co \Ty*(\Gamma)\), there are maps
  \begin{alignat*}{5}
    &\Fib*[\PSh*{\cat{C}_0}](\Gamma, A) &&\longrightarrow \Fib*[\PSh*{\cat{C}}](\PM{R}\Gamma, \PM{R} A),
    \quad
    &\alpha &\longmapsto \PM{R}\alpha, \\
    &\TrivFib*[\PSh*{\cat{C}_0}](\Gamma, A) &&\longrightarrow \TrivFib*[\PSh*{\cat{C}}](\PM{R}\Gamma, \PM{R} A),
    \quad
    &\alpha &\longmapsto \PM{R}\alpha.
  \end{alignat*}
\end{lemma}
The composite \(\PM{E} \coloneqq \PM{R}\PM{U}\) induces a monad on \(\PSh*{\cat{C}}\).
As such, it is pointed by \(\eta_\Gamma \co \Gamma \to \PM{R} \Gamma\) and therefore \(\LO{E}\) is a lex operation.
Furthermore, \(\LO{E}A = (\PM{R} (\PM{U} A))\eta_\Gamma = \LO{R}(\PM{U} A)\).
As a corollary of \cref{prop:cofree-presheaf:fibrancy}, the pseudomorphism \(\PM{E}\) and the induced lex operation \(\LO{E}\) lift to the model of fibrant types \(\PSh*{\cat{C}}_{\Fib*}\).

\subsection{The cobar descent data operation}

We now define the \emph{cobar} pointed pseudomorphism \(\PM{D}\) and prove the associated lex operation \(\LO{D}\) to be a descent data operation.
This refines the construction in~\cite[Definition~17]{CoquandRuchSattler2021} and is similar to the construction by Shulman~\cite[Section~7]{Shulman2019}.
In contrast to the construction by Shulman, this operation will be a lex modality on an intermediate model of \(\HoTT\).

The functor \(\PM{D}\) is defined as a weighted limit of a semisimplicial diagram induced by our generalized version of \(\sE\).
We fix a cosemisimplicial type \(\sP_n\) of \emph{weights}.
An element \(w \co \sP_n\) is given by a sequence \(i \co \I^{n+1}\) satisfying \(\bigvee_j i_j = 1\).
The \(k\)-th face map \(\face[n]{k} \co \sP_n \to \sP_{n+1}\) adds a \(0\) in the \(k\)-th component.

\begin{definition}
  Given a context \(\Gamma \co \PSh*{\cat{C}}\), we set
  \[
    \paren[]{\PM{D}\Gamma}_x
    \coloneqq
    \set[\bigg]{
      \gamma \co \prod_{n \geq 0} (\PM{E}^{n+1}\Gamma)^{\sP_n}_x
    }[
      \gamma_{n+1} \cc \face{k} = \paren{\PM{E}^k\eta}_x \cc \gamma_n
    ].
  \]
  On morphisms, \(\PM{D}\) is given pointwise in terms of \(\PM{E}\).
  The unit \(\tau_\Gamma \co \Gamma \to \PM{D} \Gamma\) is given by \((\tau_\Gamma\gamma)_n(i) \coloneqq \eta^{n+1}(\gamma)\).
  Given \(A \co \Ty*(\Gamma)\), we set
  \[
    (\PM{D} A)\gamma
    \coloneqq
    \set[\Bigg]{
      a \co \prod_{\substack{ n \geq 0 \\ i \co \sP_n }}  (\PM{E}^{n+1}A)(x, \gamma_n i)
    }[
      a_{n+1} \cc \face{k} = \paren{\PM{E}^k\eta}_{(x, \gamma)} \cc a_n
    ].
  \]
  For \(a \co \El*(\Gamma, A)\), we set \(\paren[\big]{(\PM{D} a)\gamma}_n(i) \coloneqq (\PM{E}^{n+1}a)(x, \gamma_ni)\).
\end{definition}

Like \(\PM{E}\) and \(\PM{U}\), the pointed pseudomorphism \(\PM{D}\) has certain strict preservation properties.
These are crucial for proving that it possesses certain homotopical properties.

\begin{lemma}\label[lemma]{prop:cobar-pseudomorphism:preservation-properties}
  The pseudomorphism \(\PM{D}\) preserves products and path types, and is therefore in particular left exact.
\end{lemma}
\begin{proof}
  As in \cref{prop:cofree-presheaf:preservation-properties}, it suffices to show preservation of products.
  This follows from \(E\) preserving products by \cref{prop:cofree-presheaf:preservation-properties} and the fact that products commute with limits.
\end{proof}

\subsubsection{Fibrancy and rectification of equivalences}

The pseudomorphism \(\PM{D}\) strictifies levelwise (trivial) fibrancy structures.
In particular, it lifts to the model of fibrant types.
The argument from~\cite{CoquandRuchSattler2021} translates to our setting.

\begin{lemma}\label[lemma]{prop:cobar-pseudomorphism:preserves-fibrancy}
  Naturally in \(\Gamma \co \PSh*{\cat{C}}\), given \(A \co \Ty*(\Gamma)\) there are maps
  \begin{alignat*}{5}
    &\Fib*[\PSh*{\cat{C}_0}](\PM{U}\Gamma, \PM{U}A) &&\longrightarrow \Fib*[\PSh*{\cat{C}}](\PM{D}\Gamma, \PM{D} A),
    \quad
    &\alpha &\longmapsto \PM{D}\alpha, \\
    &\TrivFib*[\PSh*{\cat{C}_0}](\PM{U}\Gamma, \PM{U}A) &&\longrightarrow \TrivFib*[\PSh*{\cat{C}}](\PM{D}\Gamma, \PM{D} A),
    \quad
    &\alpha &\longmapsto \PM{D}\alpha.
  \end{alignat*}
\end{lemma}

A crucial consequence of the preservation of trivially fibrant types is that \(\PM{D}\) strictifies levelwise equivalences between fibrant types.
We strengthen the result from~\cite[Corollary~22]{CoquandRuchSattler2021} slightly to levelwise equivalences between levelwise fibrant types.
This is essentially an application of Ken Brown's lemma~\cite[Factorization Lemma]{Brown1973} in our setting.

\begin{proposition}\label[proposition]{prop:cobar-strictifies-levelwise-equivalences}
  Naturally in \(\Gamma \co \PSh*{\cat{C}}\), given \(A, B \co \Ty*^{\Fib*_0}(\Gamma)\), \(f \co A \to B\), there is a map
  \[
    \El*_{\PSh*{\cat{C}_0}}(\PM{U}\Gamma, \IsEquiv(\PM{U}f)) \longrightarrow \El*_{\PSh*{\cat{C}}}(\PM{D} \Gamma, \IsEquiv(\PM{D} f)).
  \]
\end{proposition}

\subsubsection{Descent data operation structure}

The pointed pseudomorphism \(\PM{D}\) induces a lex operation \(\LO{D}\) with unit \(\tau_A \co A \to \LO{D}A\).
We show that \(\LO{D}\) is a descent data operation on \(\PSh*{\cat{C}}_{\Fib*}\) following~\cite{CoquandRuchSattler2021}.

\begin{lemma}\label[lemma]{prop:cobar-unit-is-levelwise-equivalence}
  Naturally in \(\Gamma \co \PSh*{\cat{C}}\), given \(A \co \Ty*(\Gamma)\), there is a map
  \(
    \Fib*_{\PSh*{\cat{C}_0}}(\PM{U}\Gamma, \PM{U}A) \to \El*_{\PSh*{\cat{C}_0}}(\PM{U}\Gamma, \IsEquiv(\PM{U}\tau_A)).
  \)
\end{lemma}

Using the above, we show that both \(\LO{D}\tau_A\) and \(\tau_{\LO{D} A}\) are equivalences.
In turn, this implies that \(\LO{D}\) is a descent data operation.

\begin{corollary}\label[corollary]{prop:cobar-lex-operation:elimination-unique}
  Naturally in \(\Gamma\), given \(A \co \Ty*_{\PSh*{\cat{C}}}(\Gamma)\), there is a map
  \(
    \Fib*_{\PSh*{\cat{C}_0}}(\PM{U}\Gamma, \PM{U}A) \to \El*_{\PSh*{\cat{C}}}(\Gamma, \IsEquiv(\LO{D}\tau_A)).
  \)
\end{corollary}
\begin{proof}
  By \cref{prop:cobar-strictifies-levelwise-equivalences,prop:cobar-unit-is-levelwise-equivalence}.
\end{proof}

The second claim is proven again following~\cite{CoquandRuchSattler2021} by exhibiting a zig-zag of homotopies \(\LO{D}\tau_A \sim \tau_{\LO{D} A}\) for \(A\co \Ty*_{\PSh*{\cat{C}}}^{\Fib*_0}(\Gamma)\).
We remark that it again suffices to assume that \(A\) is levelwise fibrant since \(\LO{D} A\) and \(\LO{D}^2 A\) are then fibrant by \cref{prop:cobar-pseudomorphism:preserves-fibrancy}.

\begin{lemma}\label[lemma]{prop:cobar-lex-operation:free-algebras-exist}
  Naturally in \(\Gamma\), given \(A \co \Ty*_{\PSh*{\cat{C}}}(\Gamma)\), there is a map \(\Fib*_{\PSh*{\cat{C}_0}}(\PM{U}\Gamma, \PM{U}A) \to \El*_{\PSh*{\cat{C}}}(\Gamma, \IsEquiv(\tau_{\LO{D} A}))\).
\end{lemma}

\subsection{Core properties of \texorpdfstring{\(\LO{D}\)}{D}}\label[subsection]{sec:cobar:core-properties}

We record a number of useful results about \(\LO{D}\).
In general, working with \(\LO{D}\)-modal types allows us to strictify levelwise into coherent structures.
The key property of \(\LO{D}\), which again translates from~\cite[Section~5.2, Theorem]{CoquandRuchSattler2021}, is: between modal types, levelwise equivalences and equivalences coincide.

\begin{proposition}\label[proposition]{prop:cobar-strictifies-equivalences-between-levelwise-modal-types}
  Naturally in \(\Gamma \co \PSh*{\cat{C}}\), given \(A, B \co \Ty*_{\PSh*{\cat{C}}}^{\LO{D}}(\Gamma)\) and \(f \co A \to B\), there is a logical equivalence
  \(
    \El*_{\PSh*{\cat{C}_0}}(\PM{U}\Gamma, \IsEquiv(\PM{U}f)) \longleftrightarrow \El*_{\PSh*{\cat{C}}}(\Gamma, \IsEquiv(f)).
  \)
\end{proposition}
\begin{proof}
  Assume \(U f\) is an equivalence.
  By \cref{prop:cobar-strictifies-levelwise-equivalences}, the map \(\LO{D} f\) is an equivalence.
  Then \(f\) is an equivalence by \(2\)-out-of-\(3\) as \(\LO{D} f \cc \tau_A = \tau_B \cc f\) by naturality of the unit.
\end{proof}

As a consequence of strictifying trivially fibrant types, \ie, \((-2)\)-truncated types, we have that \(\LO{D}\) sends levelwise \(n\)-truncated types to \(n\)-truncated types.
This is a refinement of the usual fact that lex modalities preserve \(n\)-types for this particular modality.

\begin{proposition}\label[proposition]{prop:cobar-strictifies-n-types}
  For all \(n \geq -2\), naturally in \(\Gamma \co \PSh*{\cat{C}}\), given \(A \co \Ty*_{\PSh*{\cat{C}}}^{\Fib*}(\Gamma)\), there is a map
  \(
    \El*_{\PSh*{\cat{C}_0}}(\PM{U}\Gamma, \IsTrunc{n}(\PM{U}A))
    \to
    \El*_{\PSh*{\cat{C}}}(\Gamma, \IsTrunc{n}(\LO{D} A)).
  \)
\end{proposition}
\begin{proof}
  It suffices to show for \(n\)-truncated \(A \co \Ty*(\Gamma)\) that \(\PM{D} A \co \Ty*(\PM{D} \Gamma)\) is \(n\)-truncated.
  We proceed by induction.
  The case \(n = -2\) follows from \cref{prop:cobar-pseudomorphism:preserves-fibrancy}.
  For the inductive step, we have an element of \(\El*_{\PSh*{\cat{C}_0}}(\PM{U}\Gamma, \IsTrunc{n+1}(\PM{U}A))\).
  Hence, we obtain an element of \(\El*_{\PSh*{\cat{C}}}(\PM{D}\Gamma.\PM{D} A.\PM{D} A\p, \IsTrunc{n}(\PM{D}\Id_A))\) by induction.
  The goal follows since \(\PM{D}\) preserves path types by \cref{prop:cobar-pseudomorphism:preservation-properties}.
\end{proof}

\section{Lifting descent data operations to presheaves}\label[section]{sec:lifting-of-descent-data-operations}

We work externally unless stated otherwise.
In this section, we show how to lift a descent data operation from \(\CubicalSet_{\Fib}\) to a descent data operation on \(\PSh{\cat*{C}}_{\Fib}\).

\begin{assumption}\label[assumption]{ass:base-model-lex-operation}
  Let \(\LexOp\), \(\vartheta_A \co A \to \LexOp A\) be a lex operation on \(\cSet\) lifting to one on \(\cSet_{\Fib}\).
\end{assumption}
This lex operation on \(\cSet\) can be lifted to a lex operation \(\LexOp[\cat{C}]\) on \(\PSh{\cat*{C}}\).

\begin{construction}[Internally to \(\CubicalSet\)]\label[construction]{prop:lifting-of-lex-operations-over-fibrant-sites:lex-operation-on-presheaves-exists}
  On contexts \(\Delta, \Gamma \co \PSh*{\cat{C}}\), substitutions \(\sigma \co \Delta \to \Gamma\), types \(A \co \Ty*(\Gamma)\), and elements \(a \co \El*(\Gamma, A)\) the pseudomorphism \(\PM{T}_{\cat{C}}\) with unit \(\vartheta\) is given by
  \begin{gather*}
    \paren{\PM{T}_{\cat{C}}\Gamma}_x \coloneqq \LexOp\Gamma_x,
    \quad
    \paren{\PM{T}_{\cat{C}}\sigma}_x \coloneqq \LexOp\sigma_x,
    \quad
    (\vartheta_\Gamma)_x\gamma \coloneqq \vartheta_{\Gamma_x}\gamma,
    \\
    (\PM{T}_{\cat{C}}A)(x, \gamma) \coloneqq (\LexOp* A)(x, \gamma),
    \quad
    \paren{\PM{T}_{\cat{C}}a}(x, \gamma) \coloneqq \paren{\widetilde{\LexOp}a}(x, \gamma).
  \end{gather*}
\end{construction}

\begin{remark}\label[remark]{prop:lifting-of-lex-operations-over-fibrant-sites:lex-operation-on-presheaves-preserved}
  The morphism \(\PM{U} \co \PSh*{\cat{C}} \to \PSh*{\cat{C}_0}\) preserves the above lifted lex operation, in the sense that \(\PM{U}(\PM{T}_{\cat{C}} A) = \PM{T}_{\cat{C}_0} (\PM{U} A)\), \(\PM{U}(\vartheta_A) = \vartheta_{\PM{U}A}\), \etc, since it is given by precomposition.
\end{remark}

We furthermore assumed that \(\sT\) lifts to the inner model \(\cSet_{\Fib}\).
Even if \(A \in \Ty_{\PSh{\cat*{C}}}(\Gamma)\) is fibrant, we cannot expect \(\PM{T}_{\cat{C}}A\) to be fibrant over \(\PM{T}_{\cat{C}}\Gamma\).
To partially bridge this gap (cf.~\cref{prop:lifting-of-lex-operations-over-fibrant-sites:lex-operation-on-presheaves-preserves-levelwise-fibrancy}), we need a condition on \(\cat{C}\).

\begin{assumption}\label[assumption]{ass:fibrant-cat}
  We assume that \(\cat{C}_1\) is a fibrant family of types over \(\cat{C}_0 \times \cat{C}_0\).
\end{assumption}

In \cref{sec:lifting-of-descent-data-operations:fibrant-sites}, we show that with this additional assumption, if \(A \in \Ty_{\PSh{\cat*{C}}}(\Gamma)\) is levelwise fibrant, then so is \(\PM{T}_{\cat{C}}A \in \Ty_{\PSh{\cat*{C}}}(\PM{T}_{\cat{C}}\Gamma)\).
We then consider the composite \(\PM{D}\PM{T}_{\cat{C}}\) of \(\PM{T}\) with the cobar pseudomorphism.
By \cref{prop:cobar-pseudomorphism:preserves-fibrancy}, this yields in particular a lex operation on \(\PSh{\cat*{C}}_{\Fib}\).
In \cref{sec:lifting-of-descent-data-operations:descent-data-operations}, we show that \(\LO{D}\LO{T}_{\cat{C}}\) is a descent data operation if \(\LO{T}\) is and characterize its modal types.
The inner model \(\PSh{\cat*{C}}_{\LO{D}\LO{T}}\) inherits properties of \(\cSet_{\LO{T}}\).
In particular, in \cref{sec:lifting-of-descent-data-operations:properties} we show that it inherits dependent choice.

\subsection{Levelwise strengthening of fibrancy}\label[subsection]{sec:lifting-of-descent-data-operations:fibrant-sites}

Internally to \(\cSet\), fibrancy of \(B \co \Ty*_{\PSh*{\cat{C}_0}}(\Delta)\) means we have \(\Fib*(\Delta_x, B_x)\) for all \(x \co \cat{C}_0\).
This does not imply \(\Fib*(\Sigma_{\cat{C}_0}\Delta, B)\)~\cite[Remark~5.9]{OrtonPitts2018}.
However, in the case where \(\Delta = \PM{U}\Gamma\) and \(B = \PM{U}A\), the following shows this implication under \cref{ass:fibrant-cat}.

\begin{theorem}[Internal to \(\cSet\)]\label[theorem]{prop:fibrant-sites:strengthen-fibrancy}
  Given a fibrancy structure\footnote{Since the cofibration used in the proof is empty, a transport structure actually suffices for this claim.} \(\Fib*(\cat{C}_0 \times \cat{C}_0, \cat{C}_1)\), naturally in \(\Gamma \co \PSh*{\cat{C}}\), for \(A \co \Ty*_{\PSh*{\cat{C}}}(\Gamma)\) there is a map
  \(
    \Fib*[\PSh*{\cat{C}_0}](\PM{U}\Gamma, \PM{U}A) \to \Fib*[\PSh*{\cat{C}_0}](\Sigma_{\cat{C}_0}\PM{U}\Gamma, \PM{U}A)
  \).
\end{theorem}
\begin{proof}
  Let \(\kappa \co \Fib*(\cat{C}_0 \times \cat{C}_0, \cat{C}_1)\) and \(\alpha \co \Fib*_{\PSh*{\cat{C}_0}}(\PM{U}\Gamma, \PM{U}A)\).
  Since our interval has connections, it suffices to give a composition structure~\cite[Section~4.4]{CohenCoquandHuberMortberg2015}.
  Consider a composition problem \(e \co \set{0,1}\), \(\varphi \co \Cof\), \(x \co \I \to \cat{C}_0\), \(\gamma \co \prod_{i \co \I} \Gamma_{x_i}\), and \(a \co \prod_{i \co \I} [i = e \vee \varphi] \to A\gamma_i\).
  Set \(\overline{e} \coloneqq 1 - e\).
  Our goal is to construct an element of \(A\gamma_{\overline{e}}[\varphi \mapsto a_{\overline{e}}]\).
  Define \(f \coloneqq \kappa(\overline{e}, \bot, \pair{x_{\overline{e}}, x}, \id_{x_{\overline{e}}}) \co \prod_{i \co \I} \cat{C}_1(x_{\overline{e}}, x_i)\), which satisfies by construction \(f_{\overline{e}} = \kappa(\overline{e}, \bot, \pair{x_{\overline{e}}, x}, \id_{x_{\overline{e}}}, \overline{e}) = \id_{x_{\overline{e}}}\).
  We have \(f_i \co \cat{C}_1(x_{\overline{e}}, x_i)\).
  Using the restriction operations for \(\Gamma\) and \(A\), define
  \[
    (\gamma f)_i \coloneqq \gamma_i f_i \co \I \to \Gamma_{x_{\overline{e}}}, \quad
    (a f)_i \coloneqq a_if_i \co \prod_{i \co \I} [i = e \vee \varphi] \to A\paren[\big]{x_{\overline{e}}, (\gamma f)_i}[i = e \vee \varphi \mapsto a_if_i].
  \]
  These satisfy \((\gamma f)_{\overline{e}} = \gamma_{\overline{e}}\id_{x_{\overline{e}}} = \gamma_{\overline{e}}\) and \((af)_{\overline{e}} = a_{\overline{e}} \id_{x_{\overline{e}}} = a_{\overline{e}}\) by construction.
  Hence, we can define a composition operation \(\overline{\alpha}\) by
  \[
    \overline{\alpha}(e, \varphi, \pair{x, \gamma}, a) \coloneqq \alpha_{x_{\overline{e}}}(e, \varphi, \gamma f, a f, \overline{e}) \co A\gamma_{\overline{e}}[\varphi \mapsto a_{\overline{e}}].
  \]

  We claim that this construction is natural in \(\Gamma\).
  Consider \(\sigma \co \Delta \to \Gamma\) in \(\PSh*{\cat{C}}\) and a composition problem \(\varphi \co \Phi\), \(x \co \I \to \cat{C}_0\), \(\delta \co \prod_{i \co \I} \Delta_{x_i}\), \(a \co \prod_{i \co \I} [i = e \vee \varphi] \to A\sigma_{x_i}\delta_i\).
  We view \(\sigma\) as a map \(\cat{C}_0.\Delta \to \cat{C}_0.\Gamma, (x, \delta) \mapsto (x, \sigma_x\delta)\).
  For \(f\co \prod_{i \co \I} \cat{C}_1(x_{\overline{e}}, x_i)\) defined as above, we have
  \[
    \sigma_{x_{\overline{e}}}(\delta f)
    = \lambda_i.\sigma_{x_{\overline{e}}}(\delta f)_i
    = \lambda_i.\sigma_{x_{\overline{e}}}(\delta_if_i)
    = \lambda_i.(\sigma_{x_i}\delta_i)f_i
    = (\lambda_i.\sigma_{x_i} \delta_i) f.
  \]
  Writing \((\sigma \delta)_i \coloneqq \sigma_{x_i} \delta_i\) we thus have \(\sigma(\delta f) = (\sigma \delta) f\), and therefore
  \begin{multline*}
    \overline{\alpha\sigma}(e, \varphi, \pair{x, \delta}, a)
    = \alpha\sigma_{x_{\overline{e}}}(e, \varphi, \delta f, a f, \overline{e})
    = \alpha_{x_{\overline{e}}}\sigma_{x_{\overline{e}}}(e, \varphi, \delta f, a f, \overline{e})
    = \alpha_{x_{\overline{e}}}(e, \varphi, \sigma(\delta f), a f, \overline{e}) \\
    = \alpha_{x_{\overline{e}}}(e, \varphi, (\sigma \delta)f, a f, \overline{e})
    = \overline{\alpha}(e, \varphi, \pair{x, \sigma \delta}, a)
    = \overline{\alpha}(e, \varphi, \sigma\pair{x, \delta}, a)
    = \overline{\alpha}\sigma(e, \varphi, \pair{x, \delta}, a). \qedhere
  \end{multline*}
\end{proof}

\begin{remark}\label[remark]{prop:lifting-of-lex-operations-over-fibrant-sites:weaken-fibrancy}
  By \cref{prop:fibrant-sites:strengthen-fibrancy}, we can see the pseudomorphism \(U\) as factoring over a model on \(\PSh{\cat*{C}_0}\) where we take the types in context \(\Gamma\) to be \(\Ty_{\cSet}^{\Fib}(\cat{C}_0.\Gamma)\).
  We write \(\cSet_{\Fib} \comma \cat{C}_0\) for this model because it is, up to equivalence of the category of contexts, the slice model of \(\cSet_{\Fib}\) over \(\cat{C}_0\).
  Note that every fibrant type of \(\cSet_{\Fib} \comma \cat{C}_0\) is also one of \(\PSh{\cat*{C}_0}_{\Fib}\).
\end{remark}

Combining \cref{prop:fibrant-sites:strengthen-fibrancy} with \cref{ass:base-model-lex-operation}, stating that \(\LO{T}\) lifts to fibrant types on \(\cSet\), we obtain that \(\PM{T}_{\cat{C}}\) lifts to levelwise fibrant types.

\begin{lemma}\label[lemma]{prop:lifting-of-lex-operations-over-fibrant-sites:lex-operation-on-presheaves-preserves-levelwise-fibrancy}
  Naturally in \(\Gamma \in \PSh{\cat*{C}}\), given \(A \in \Ty(\Gamma)\), we have a function from 
  \(
    \Fib[\PSh{\cat*{C}_0}](\PM{U}\Gamma, \PM{U}A)
  \)
  to
  \(
    \Fib[\PSh{\cat*{C}_0}](\PM{U}(\PM{T}_{\cat{C}}\Gamma), \PM{U}({\PM{T}_{\cat{C}}} A))
  \).
\end{lemma}
\begin{proof}
  We have the following chain of natural logical implications:
  \begin{align*}
    \Fib[\PSh{\cat*{C}_0}](\PM{U}\Gamma, \PM{U}A)
    &\longrightarrow \Fib[\CubicalSet](\cat{C}_0.\PM{U}\Gamma, \PM{U}A) \tag{\cref{prop:fibrant-sites:strengthen-fibrancy}} \\
    &\longrightarrow \Fib[\CubicalSet](\cat{C}_0.\LexOp(\PM{U}\Gamma), \LexOp* (\PM{U} A)) \tag{\cref{ass:base-model-lex-operation}} \\
    &\longrightarrow \Fib[\PSh{\cat*{C}_0}](\LexOp[\cat{C}_0](\PM{U}\Gamma), \LexOp[\cat{C}_0] (\PM{U}A)) \tag{\cref{prop:lifting-of-lex-operations-over-fibrant-sites:lex-operation-on-presheaves-exists}, \cref{prop:lifting-of-lex-operations-over-fibrant-sites:weaken-fibrancy}}
  \end{align*}
  We conclude using \cref{prop:lifting-of-lex-operations-over-fibrant-sites:lex-operation-on-presheaves-preserved}.
\end{proof}

Similar to preservation of fibrancy, the descent data operation structure on \(\LexOp\) only lifts levelwise.
Hence, when lifting \(\PM{T}\) to presheaves, we only consider the composite with \(\PM{D}\).

\begin{lemma}\label[lemma]{prop:lifting-of-lex-operations-over-fibrant-sites:lex-operation-on-presheaves-descent-data-operation-lifts-levelwise}
  If \(\LexOp\) is a descent data operation on \(\cSet_{\Fib}\), then naturally in \(\Gamma \in \PSh{{\cat*{C}}}\), the maps \(\LexOp[\cat{C}] \vartheta_A, \vartheta_{\LexOp[\cat{C}] A} \co \LexOp[\cat{C}] A \to \LexOp_{\cat{C}}^2 A\) are levelwise equivalences for all \(A \in \Ty_{\PSh{{\cat*{C}}}}^{\Fib_0}(\Gamma)\).
\end{lemma}
\begin{proof}
  Since \(\cat{C}\) is fibrant, we have \(\PM{U}A \in \Ty_{\CubicalSet}^{\Fib}(\cat{C}_0.\Gamma)\) by \cref{prop:fibrant-sites:strengthen-fibrancy}. 
  Because \(\LexOp\) is a descent data operation, the maps \(\LexOp\vartheta_{\PM{U}A}, \vartheta_{\LexOp \PM{U}A} \co \LexOp A \to \LexOp^2 A\) are equivalences in the sense of \(\CubicalSet \comma \cat{C}_0\).
  By construction, these coincide with \(\PM{U}\paren{\LexOp[\cat{C}] \vartheta_A}, \PM{U}\paren{\vartheta_{\LexOp[\cat{C}] A}} \co \PM{U}\paren{\LexOp[\cat{C}] A} \to \PM{U}\paren{\LexOp_{\cat{C}}^2 A}\).
\end{proof}

Another useful fact is that \(\LexOp\) preserves levelwise equivalences.

\begin{lemma}\label[lemma]{prop:lifting-of-lex-operations-over-fibrant-sites:lex-operation-on-presheaves-preserves-levelwise-equivalences-between-levelwise-fibrant-types}
  Naturally in \(\Gamma \in \PSh{\cat*{C}}\), given \(A, B \in \Ty^{\Fib_0}_{\PSh{\cat*{C}}}(\Gamma)\) and \(f \co A \to B\), there is a function
  \(
    \El_{\PSh{\cat*{C}_0}}(\PM{U}\Gamma, \IsEquiv(\PM{U}f))
    \to
    \El_{\PSh{\cat*{C}_0}}(\PM{U}\Gamma, \IsEquiv(\PM{U}(\LexOp[\cat{C}]f)))
  \).
\end{lemma}
\begin{proof}
  Using \cref{prop:fibrant-sites:strengthen-fibrancy}, we view \(\PM{U}f \co \PM{U}A \to \PM{U}B\) as an equivalence between fibrant types in \(\CubicalSet_{\Fib} \comma \cat{C}_0\) (\ie, in context \(\cat{C}_0.\Gamma\) of \(\cSet_{\Fib}\)).
  Hence, the map \(\LexOp (\PM{U} f) \co \LexOp (\PM{U} A) \to \LexOp (\PM{U} A)\) is an equivalence since every lex operation preserves equivalences by path induction (see~\cite[Section 2.2, Theorem]{CoquandRuchSattler2021}).
  This map corresponds to \(\LexOp[\cat{C}_0](\PM{U} f) \co \LexOp[\cat{C}_0](\PM{U} A) \to \LexOp[\cat{C}_0](\PM{U} B)\) in \(\PSh{\cat{C}_0} \comma \PM{U}\Gamma\).
  By \cref{prop:lifting-of-lex-operations-over-fibrant-sites:lex-operation-on-presheaves-preserved}, \(\PM{U}(\LexOp[\cat{C}]f) = \LexOp[\cat{C}_0](\PM{U} f)\), so the claim follows.
\end{proof}

\subsection{Lifting of descent data operations}\label[subsection]{sec:lifting-of-descent-data-operations:descent-data-operations}

Internally to \(\CubicalSet\), we can construct the composite pseudomorphism \(\PM{D} \circ \PM{T}_{\cat{C}}\).
We will omit the subscript \(\cat{C}\) if it is clear from context, such as when we compose with \(\PM{D}\).
The unit for this composite is given by
\(
  \tau_{\PM{T} \Gamma} \circ \vartheta_\Gamma = \PM{D}\vartheta_\Gamma \circ \tau_\Gamma \co \Gamma \to \PM{DT} \Gamma.
\)
This yields a lex operation \(\LO{DT}\) on \(\PSh{\cat*{C}}\) with unit \(\pi_{A} \co A \to \LO{DT}A\) given by \(\tau_{\LO{T} A} \circ \vartheta_A = \LO{D}\vartheta_A \circ \tau_A\).

\begin{assumption}\label{ass:base-model-descent-data-operation}
  The lex operation \(\LexOp\) on \(\cSet\) is a descent data operation.
\end{assumption}

Under the above assumption, we show that \(\LO{D}\LexOp\) is a descent data operation on \(\PSh{\cat*{C}}_{\Fib}\).
In particular, the composite lex operation already lifts to \(\PSh{\cat*{C}}_{\Fib}\) by \cref{prop:lifting-of-lex-operations-over-fibrant-sites:lex-operation-on-presheaves-preserves-levelwise-fibrancy,prop:cobar-pseudomorphism:preserves-fibrancy}.

\begin{proposition}\label[proposition]{prop:lifting-of-lex-operations-over-fibrant-sites:composite-lex-operation-exists}
  Naturally in \(\Gamma \in \PSh{\cat*{C}}\), there is for all \(A \in \Ty_{\PSh{\cat*{C}}}(\Gamma)\) a function from
  \(
    \Fib[\PSh{\cat*{C}_0}](\PM{U}\Gamma, \PM{U}A)
  \)
  to
  \(\
  \Fib[\PSh{\cat*{C}}]({\PM{DT} }\Gamma, {\PM{DT} } A).
  \)
\end{proposition}

To show that \(\LO{D}\LO{T}\) is also a descent data operation, we show that \(\pi_{\sD\LexOp  A}, \sD\LexOp \pi_A \co \LO{D}\LO{T}  A \to (\LO{D}\LO{T})^2 A\) are equivalences.
In particular, we show this if \(A\) is only levelwise fibrant.

\begin{lemma}\label[lemma]{prop:lifting-of-lex-operations-over-fibrant-sites:elimination-unique}
  Naturally in \(\Gamma \in \PSh{\cat*{C}}\), for all \(A \in \Ty_{\PSh{\cat*{C}}}^{\Fib_0}(\Gamma)\) the map \(\sD\LexOp \pi_A \co \sD\LexOp A \to (\sD\LexOp) ^2A\) has the structure of an equivalence.
\end{lemma}
\begin{proof}
  Our goal is to show that the diagonal in the following defining square is an equivalence:
  \begin{center}\(\begin{tikzcd}[column sep=3em]
    {\sD \LexOp A}
    \ar[r, "\sD\LexOp \tau_A"]
    \ar[d, "\sD\LexOp \vartheta_A"']
  &
    {\sD \LexOp \sD A}
    \ar[d, "\sD\LexOp\sD\vartheta_A"]
  \\
    {\sD \LexOp^2 A}
    \ar[r, "\sD\LexOp\tau_{\LexOp A}"]
  &
    {(\sD \LexOp)^2 A\rlap{.}}
  \end{tikzcd}\)\end{center}
  By \cref{prop:lifting-of-lex-operations-over-fibrant-sites:lex-operation-on-presheaves-preserves-levelwise-fibrancy,prop:lifting-of-lex-operations-over-fibrant-sites:composite-lex-operation-exists}, the types \(\LexOp A\),  \(\LexOp^2A\), and \(\LexOp\sD\LexOp A\) are levelwise fibrant.
  Hence, it suffices to show by \cref{prop:cobar-strictifies-levelwise-equivalences} that the composite \(\LexOp A \xrightarrow{\LexOp \vartheta_A} \LexOp^2 A \xrightarrow{\LexOp \tau_{\LexOp A}} \LexOp \sD \LexOp A\) is levelwise an equivalence.
  \(\LexOp\vartheta_A\) is a levelwise equivalence by \cref{prop:lifting-of-lex-operations-over-fibrant-sites:lex-operation-on-presheaves-descent-data-operation-lifts-levelwise}.
  Since \(\LexOp A\) is levelwise fibrant, the map \(\tau_{\LexOp A}\) is a levelwise equivalence by \cref{prop:cobar-unit-is-levelwise-equivalence}.
  By \cref{prop:lifting-of-lex-operations-over-fibrant-sites:lex-operation-on-presheaves-preserves-levelwise-equivalences-between-levelwise-fibrant-types}, \(\LexOp\) preserves the levelwise equivalence.
\end{proof}

\begin{lemma}\label[lemma]{prop:lifting-of-lex-operations-over-fibrant-sites:free-algebras-exist}
  Naturally in \(\Gamma \in \PSh{\cat*{C}}\), for all \(A \in \Ty_{\PSh{\cat*{C}}}^{\Fib_0}(\Gamma)\) the map \(\pi_{\sD\LexOp A} \co \sD\LexOp A \to (\sD\LexOp)^2A\) has the structure of an equivalence.
\end{lemma}
\begin{proof}
  In the following diagram, the left square commutes by naturality of \(\vartheta\) and the right square is the defining square for \(\pi_{\sD\LexOp A}\):
  \begin{center}\(\begin{tikzcd}[column sep=3em]
    {\LexOp A}
    \ar[r, "\tau_{\LexOp A}"]
    \ar[d, "\vartheta_{\LexOp A}"']
  &
    {\sD \LexOp A}
    \ar[r, "\tau_{\sD \LexOp A}"]
    \ar[d, "\vartheta_{\sD \LexOp A}"]
  &
    {\sD^2 \LexOp A}
    \ar[d, "\sD\vartheta_{\sD \LexOp A}"]
  \\
    {\LexOp^2 A}
    \ar[r, "\LexOp \tau_{\LexOp A}"]
  &
    {\LexOp \sD \LexOp A}
    \ar[r, "\tau_{\LexOp \sD \LexOp A}"]
  &
    {(\sD \LexOp)^2 A \rlap{.}}
  \end{tikzcd}\)\end{center}
  By \cref{prop:lifting-of-lex-operations-over-fibrant-sites:lex-operation-on-presheaves-preserves-levelwise-fibrancy,prop:cobar-pseudomorphism:preserves-fibrancy}, all types above are levelwise fibrant.
  All horizontal maps are levelwise equivalences by \cref{prop:cobar-unit-is-levelwise-equivalence,prop:lifting-of-lex-operations-over-fibrant-sites:lex-operation-on-presheaves-preserves-levelwise-equivalences-between-levelwise-fibrant-types}.
  The vertical map on the left is a levelwise equivalence by \cref{prop:lifting-of-lex-operations-over-fibrant-sites:lex-operation-on-presheaves-descent-data-operation-lifts-levelwise}.
  Hence, by \(2\)-out-of-\(3\) for levelwise equivalences between levelwise fibrant types, all maps above are levelwise equivalences.

  The types \(\sD \LexOp A\) and \((\sD \LexOp)^2 A\) are fibrant by \cref{prop:lifting-of-lex-operations-over-fibrant-sites:lex-operation-on-presheaves-preserves-levelwise-fibrancy,prop:cobar-pseudomorphism:preserves-fibrancy}, and \(\sD\)-modal by \cref{prop:cobar-lex-operation:free-algebras-exist}.
  The claim follows since levelwise equivalences and equivalences between \(\sD\)-modal types coincide by \cref{prop:cobar-strictifies-equivalences-between-levelwise-modal-types}.
\end{proof}

We can now show that \(\LO{D}\LO{T}\) is a descent data operation and characterize its modal types.
As a corollary, we obtain an inner model \(\PSh{\cat*{C}}_{\LO{D}\LO{T}}\) of \(\LO{D}\LO{T}\)-modal types in \(\PSh{\cat*{C}}_{\Fib}\).

\begin{theorem}\label[theorem]{prop:lifting-of-lex-operations-over-fibrant-sites:modal-type-characterization}
  The lex operation \(\sD\LexOp\) on \(\PSh{\cat*{C}}_{\Fib}\) is a descent data operation.
  Furthermore, naturally in \(\Gamma \in \PSh{\cat*{C}}\), for all \(A \in \Ty^{\Fib}_{\PSh{\cat*{C}}}(\Gamma)\), there is a logical equivalence
  \[
    \El_{\PSh{\cat*{C}}}(\Gamma, \IsEquiv(\pi_A))
    \longleftrightarrow
    {
      \El_{\PSh{\cat*{C}}}(\Gamma, \IsEquiv(\tau_A))
      \times
      \El_{{\CubicalSet \comma \cat{C}_0}}(\PM{U}\Gamma, \IsEquiv(\vartheta_{\PM{U}A}))
    }.
  \]
\end{theorem}
\begin{proof}
  The lex operation lifts to the model of fibrant types by \cref{prop:lifting-of-lex-operations-over-fibrant-sites:composite-lex-operation-exists}.
  To show that \(\sD\LexOp\) is a descent data operation, we have to show that \(\pi_{\sD\LexOp A}\) and \(\sD\LexOp \pi_A\) are equivalences for all \(A \in \Ty^{\Fib}_{\PSh{\cat*{C}}}(\Gamma)\).
  This is the case by \cref{prop:lifting-of-lex-operations-over-fibrant-sites:elimination-unique,prop:lifting-of-lex-operations-over-fibrant-sites:free-algebras-exist}.

  Every \(\sD\LexOp\)-modal type is \(\sD\)-modal as every type in the image of \(\LO{D}\LexOp\) is so by \cref{prop:cobar-lex-operation:free-algebras-exist}.
  Hence, it suffices to show for a \(\sD\)-modal type \(A\) that it is \(\sD\LexOp \)-modal exactly if \(\vartheta_A \co A \to \LexOp A\) is a levelwise equivalence.
  Since the unit is given by \(\pi_A = \tau_{\LexOp A} \cc \vartheta_A \co A \to \sD\LexOp  A\), the claim follows by \(2\)-out-of-\(3\) since levelwise equivalences between \(\sD\)-modal types coincide with equivalences by \cref{prop:cobar-strictifies-equivalences-between-levelwise-modal-types}, and \(\tau_{\LexOp A}\) is a levelwise equivalence by \cref{prop:cobar-unit-is-levelwise-equivalence}.
\end{proof}

\begin{remark}\label[remark]{rem:hit-existence-in-inner-presheaves}
  If one wishes to justify higher inductive types in the model \(\PSh{\cat*{C}}_{\LO{D}\LO{T}}\), one can mirror the construction of initial algebras from \cite{CoquandRuchSattler2021,CoquandHuberMortberg2018}.
  For this, it suffices if \(\LO{T}\) is given as a parametric right adjoint.
  This is, for example, the case for the modality from \cite{Sattler2025Types}.
\end{remark}

\subsection{Properties of the model of \texorpdfstring{\(\sD\LexOp\)}{DT}-modal types}\label[subsection]{sec:lifting-of-descent-data-operations:properties}

By \cref{prop:fibrant-sites:strengthen-fibrancy}, the pseudomorphism \(\PM{U}\) and its right adjoint on types \(\LO{R}\) lift to a pseudomorphism with right adjoint on types between \(\cSet_{\Fib} \comma \cat{C}_0\) and \(\PSh{\cat*{C}}_{\Fib}\).
In this subsection, we show that it lifts further to one between \(\cSet_{\LexOp} \comma \cat{C}_0\) and \(\PSh{\cat*{C}}_{\sD\LexOp}\).
From this, we conclude that propositional truncations are computed levelwise in \(\PSh{\cat*{C}}_{\sD\LexOp}\).
This allows us to lift dependent choice from the base model.

\begin{lemma}[Internal to \(\CubicalSet\)]\label[lemma]{prop:cofree-presheaf:cobar-modal}
  Naturally in \(\Gamma \co \PSh*{\cat{C}}\), given \(A \co \Ty*^{\Fib*}_{\PSh*{\cat{C}_0}}(\PM{U}\Gamma)\), the type \(\sR A \co \Ty*^{\Fib*}_{\PSh*{\cat{C}}}(\Gamma)\) is \(\sD\)-modal.
\end{lemma}
\begin{proof}
  We have to find a left inverse of \(\tau_{\sR A} \co \sR A \to \sD(\sR A)\).
  This is equivalent to finding a map \(p \co \PM{U}(\sD(\sR A)) \to A\) with \(p \cc \PM{U} \tau_{\sR A} \sim \varepsilon_A\).
  Such a map exists since \(\PM{U}\tau_{\sR A}\) is an equivalence by \cref{prop:cobar-unit-is-levelwise-equivalence}.
\end{proof}

Next, we show that \(\LO{R}A\) is also levelwise \(\LO{T}\)-modal if \(A\) is \(\LO{T}\)-modal.

\begin{lemma}\label[lemma]{prop:cofree-presheaf:preserves-levelwise-modal-types-for-base-model-modality}
  Naturally in \(\Gamma \in \PSh{\cat*{C}}\), given  \(A \in \Ty^{\LexOp}_{\CubicalSet \comma \cat{C}_0}(\PM{U}\Gamma)\), the map \(\vartheta_{\sR A} \co \sR A \to \LexOp (\sR A)\) has a homotopy retraction.%
  \footnote{%
  Note that \(\LexOp (\sR A)\) is not necessarily fibrant, which is why do not refer to this as being \(\LexOp\)-modal.
  This preservation of \(\LexOp\)-algebras up to homotopy should be viewed as a technical, strict property of \(\sR\).%
  }
\end{lemma}
\begin{proof}
  It suffices to find a map \(p \co \PM{U} (\LexOp (\sR A)) \to A\) such that \(p \cc \PM{U}\vartheta_{\sR A} \sim \varepsilon_A\).
  By \cref{prop:lifting-of-lex-operations-over-fibrant-sites:lex-operation-on-presheaves-preserved}, we have that \(\PM{U}\vartheta_{\sR A} = \vartheta_{\PM{U  }\sR A}\).
  We have by naturality of \(\vartheta\) that \(\LexOp\varepsilon_A \cc \vartheta_{\PM{U}(\LexOp{R}A)} = \vartheta_A \cc \varepsilon_A\).
  The map \(\vartheta_A\) is by assumption an equivalence.
  Hence, \(\vartheta_A^{-1} \cc \LexOp\varepsilon_A\) is the desired map.
\end{proof}

\begin{proposition}\label[proposition]{prop:cofree-presheaf:sends-levelwise-modal-types-for-base-model-modality-to-modal-types-for-combined-modality}
  Naturally in \(\Gamma \in \PSh{\cat*{C}}\), given \(A \in \Ty^{\LO{T}}_{\PSh{\cat*{C}_0}}(\PM{U}\Gamma)\) the type \(\sR A \in \Ty^{\Fib}_{\PSh{\cat*{C}}}(\Gamma)\) is \(\sD\sT\)-modal.
\end{proposition}
\begin{proof}
  This follows from \cref{prop:lifting-of-lex-operations-over-fibrant-sites:modal-type-characterization} via \cref{prop:cofree-presheaf:cobar-modal,prop:cofree-presheaf:preserves-levelwise-modal-types-for-base-model-modality}.
\end{proof}

\begin{proposition}\label[proposition]{prop:lifting-of-lex-operations-over-fibrant-sites:levelwise-principle}
  Naturally in \(\Gamma \in \PSh{\cat*{C}}\), given \(A \in \Ty^{\sD\LexOp}(\Gamma)\), there is a logical equivalence
  \[
    \El_{\cSet_{\Fib} \comma \cat{C}_0}^{\LexOp}(\PM{U}\Gamma, \trunc{\PM{U}A})
    \longleftrightarrow
    \El_{\PSh{\cat*{C}}}^{\sD\LexOp}(\Gamma, \trunc{A}).
  \]
\end{proposition}
\begin{proof}
  First, we construct a left-to-right map.
  By \cref{prop:lifting-of-lex-operations-over-fibrant-sites:modal-type-characterization}, \(\PM{U}\trunc{A}\) is a \(\LexOp\)-modal proposition.
  Furthermore, there is a map \(\PM{U}A \to \PM{U}\trunc{A}\) and a witness that \(\PM{U}\trunc{A}\) is an h-proposition.
  Hence, we have a chain of logical implications
  \begin{align*}
    \El_{\PSh{\cat*{C}_0}}^{\LexOp}(\PM{U}\Gamma, \trunc{\PM{U}A})
    &\longrightarrow \El_{\PSh{\cat*{C}_0}}^{\LexOp}(\PM{U}\Gamma, \PM{U}\trunc{A}) \tag{universal property \(\trunc{\PM{U}A}\)} \\
    &\longrightarrow \El_{\PSh{\cat*{C}_0}}^{\LexOp}(\PM{U}\Gamma, \IsContr(\PM{U}\trunc{A})) \tag{inhabited proposition} \\
    &\longrightarrow \El_{\PSh{\cat*{C}_0}}^{\Fib}(\PM{U}\Gamma, \IsContr(\PM{U}\trunc{A})) \tag{contractible types agree} \\
    &\longrightarrow \El_{\PSh{\cat*{C}}}^{\sD}(\Gamma, \IsContr(\trunc{A})) \tag{\cref{prop:cobar-pseudomorphism:preserves-fibrancy}} \\
    &\longrightarrow \El_{\PSh{\cat*{C}}}^{\sD\LexOp}(\Gamma, \trunc{A}). \tag{center of contraction}
  \end{align*}
  Second, we construct a right-to-left map.
  By \cref{prop:cofree-presheaf:sends-levelwise-modal-types-for-base-model-modality-to-modal-types-for-combined-modality}, there is a logical equivalence
  \[
    \El_{\cSet_{\Fib} \comma \cat{C}_0}^{\LexOp}(\PM{U}\Gamma, \trunc{\PM{U}A})
    \longleftrightarrow
    \El_{\PSh{\cat*{C}}}^{\sD\LexOp}(\Gamma, \sR\trunc{\PM{U}A}).
  \]
  By composing \(\eta_A \co A \to \sR \PM{U}A\) and \(\sR\eta_{\PM{U}A} \co \sR \PM{U}A \to \sR\trunc{\PM{U}A}\), there is a map \(A \to \sR\trunc{\PM{U}A}\).
  Since \(\sR\trunc{\PM{U}A}\) is an h-proposition, the claim follows from the universal property of \(\trunc{A}\).
\end{proof}
We deduce that surjections are characterized levelwise, allowing us to lift dependent choice.

\begin{corollary}\label[corollary]{prop:lifting-of-lex-operations-over-fibrant-sites:surjections-are-levelwise}
  Naturally in \(\Gamma \in \PSh{\cat*{C}}\), given \(A, B \in \Ty^{\sD\LexOp}(\Gamma)\) and \(f \co A \to B\), there is a logical equivalence
  \[
    \El_{\PSh{\cat*{C}}}^{\sD\LexOp}(\Gamma, \IsSurj(f))
    \longleftrightarrow
    \El_{\cSet_{\Fib} \comma \cat{C}_0}^{\LexOp}(\PM{U}\Gamma, \IsSurj(\PM{U}f)).
  \]
\end{corollary}
\begin{proof}
  A map is surjective if all its fibers are merely inhabited.
  We have a chain of logical equivalences
  \begin{align*}
    \El_{\PSh{\cat*{C}}}^{\sD\LexOp}(\Gamma.B, \trunc{\Fiber{f}})
    &\longleftrightarrow
    \El_{\cSet_{\Fib} \comma \cat{C}_0}^{\LexOp}(\PM{U}\Gamma.\PM{U}B, \trunc{\PM{U}\paren{\Fiber{f}}}) \\
    &\longleftrightarrow
    \El_{\cSet_{\Fib} \comma \cat{C}_0}^{\LexOp}(\PM{U}\Gamma.\PM{U}B, \trunc{{\Fiber{\PM{U}f}}}).
  \end{align*}
  Here we use that \(\PM{U}\) preserves homotopy fibers since it preserves paths.
\end{proof}

Recall the axiom of dependent choice: every tower
\(\begin{tikzcd}[column sep=0.35cm]
  {A_0}
&
  {A_1}
  \ar[l, two heads]
&
  {\cdots}
  \ar[l, two heads]
\end{tikzcd}\)
of surjections \(f_i \co A_{i+1} \twoheadrightarrow A_i\) composes to a surjection \(\pi_0 \co \lim_{n \in \omega^\op} A_n \to A_0\).
In a model of type theory with products indexed by \(\sN\) and identity types, this limit is constructed as the dependent sum indexed by \(u \co \prod_{n : \sN} A_n\) of \(\prod_{n : \sN} f_n (u_{n+1}) \simeq_{A_n} u_{n}\).
The axiom of dependent choice is then stated as a schema: naturally in \(\Gamma\), given \(A \in \Ty(\Gamma.\sN)\), \(f \in \El(\Gamma.(n\co \sN), A_{n+1} \to A_n)\), and an element of \(\IsSurj(f) \in \Ty(\Gamma.\sN)\), we have an element of \(\IsSurj(\pi_0) \in \Ty(\Gamma)\).

\begin{proposition}
  If the model \(\CubicalSet_{\LexOp}\) satisfies dependent choice, then so does \(\PSh{\cat*{C}}_{\sD\LexOp}\).
\end{proposition}
\begin{proof}
  By \cref{prop:lifting-of-lex-operations-over-fibrant-sites:surjections-are-levelwise}, it suffices to show that \(\PM{U}\) preserves the inverse limit in the statement of dependent choice.
  Since \(\PM{U}\) preserves paths by \cref{prop:forgetful:preservation-properties}, it remains to show that it preserves \(\sN\)-indexed products.
  This follows since the type \(\sN\) in \(\PSh{\cat*{C}}_{\sD\LexOp}\) is equivalent to the \(\sD\LexOp\)-replacement of the constant presheaf at the natural numbers. 
\end{proof}

\section{Applications to synthetic mathematics}\label[section]{sec:duality}

We apply our previous results to construct models of axiom systems for synthetic mathematics.
For an equational theory, we validate in particular the duality axiom in the model \(\PSh{\cat*{C}}_{\LO{D}\LO{T}}\) for a suitable cubical category \(\cat{C}\), parametrized by a descent data operation \(\LexOp\) on \(\cSet_{\Fib}\).
This axiom is due to Blechschmidt~\cite[Theorem 4.10]{Blechschmidt2019}.
A version of the axiom (where not all quantifications are internalized to the topos) and related statements also appear in earlier work of Kock~\cite{Kock1981,Kock2014duality}.
For some applications, the model of \(\cSet_{\LexOp}\) needs to satisfy dependent choice, for which we can choose \(\LexOp\) as the lex operation defined by Sattler~\cite{Sattler2025Types}.
Our general strategy is to adapt the \(1\)-categorical arguments to our setting~\cite{Blechschmidt2019,CherubiniCoquandHutzler2024}.

\subsection{The base category}

We work internally to \(\CubicalSet\).
In this section, we define the base category for our presheaf model construction.
It should consist of models of the theory in h-sets, but still be a strict category.
Given a universe \(\Univ\), we denote the induced universe of fibrant types by \(\Univ_{\Fib*}\) and its subuniverse of \(\LexOp\)-modal types by \(\Univ_{\LexOp}\).

\subsubsection{Models}\label[subsubsection]{sec:duality:models}
We fix a theory \(\bbT\) given by a single sort, types \(O_n \co \Univ_{0, \Fib*}\) of \(n\)-ary operations for \(n \co \sN\), and universally quantified equations involving these operations indexed by types in \(\Univ_{0, \Fib*}\).

A \emph{\(\bbT\)-model} is a \(\LexOp\)-modal h-set \(A\) with a map \(o_A \co A^n \to A\) for each \(o \co O_n\) and homotopies witnessing the equations.
A morphism of models is given by a function \(h \co A \to B\) with a homotopy \(h \cc o_A \sim o_B \cc h^n\) for each \(o \co O_n\).
Models and their morphisms in a universe \(\Univ_{\LO{T}}\) form a \(\LexOp\)-modal category \(\Mod*[\bbT](\Univ_{\LexOp})\) in the sense of \(\HoTT\) (\ie, a category internal to the model of \(\LexOp\)-modal types).

In fibrant types, colimits and free \(\bbT\)-models can be constructed using quotients of h-sets and finitary inductive types.%
\footnote{These can, for example, be constructed using higher inductive types as justified by~\cite{CoquandHuberMortberg2018}.}
The submodel of \(\LexOp\)-modal types inherits these universal constructions by \(\LexOp\)-modal replacement.
We denote coproducts of \(\bbT\)-models by \(\otimes_{\bbT}\) and free models on a type \(X\) by \(\sF_{\bbT} X\), where we omit the subscript if it is clear from context.

\subsubsection{Presentations}
Let \(\UKappa\) be a subuniverse of \(\Univ_{0, \LexOp}\) of projective h-sets (\wrt~\(\LexOp\)-modal types) containing all finite types and closed under dependent sums.
Examples include the finite sets and (provided that the submodel of \(\LexOp\)-modal types satisfies countable choice) the universe of decidable subtypes of \(\sN\), which is constructively a good notion of countable type.

\begin{definition}
  A \emph{\(\UKappa\)-presentation} consists of h-sets \(X, R \co \UKappa\) with maps \(s, t \co R \rightrightarrows \sF X\).
  The presented \(\bbT\)-model \(\gen{ X }[ s \simeq t ]_{\bbT}\) is the coequalizer of the parallel maps \(\sF R \rightrightarrows \sF X\) induced by \(s,t\).
  A model is \emph{\(\UKappa\)-presentable} if it is merely equivalent to one of the above form.
  We denote the type of \(\UKappa\)-presentations by \(\Pr_{\bbT,\kappa}\) and the category of \(\UKappa\)-presentable models by \(\Mod*[\bbT,\kappa](\Univ_{\LexOp})\).
\end{definition}

\begin{remark}\label{rem:category-of-models-size}
  Up to equivalence, the category \(\Mod*[\bbT,\kappa](\Univ_{\LexOp})\) does not depend on the chosen universe \(\Univ\), as long as it is closed under quotients of h-sets and finitary inductive types.
  If \(\Univ\) also satisfies replacement (in the sense of Rijke~\cite{Rijke2017joinconstruction}) and is closed under the usual type formers, the category \(\Mod*[\bbT,\kappa](\Univ_{\LexOp})\) is \(\Univ_{\LexOp}\)-small. 
  Henceforth, we will omit the reference to \(\Univ\).
\end{remark}

We now define the base category for our model.
This will be a strictification of the above category of \(\UKappa\)-presentable models up to homotopy.
This should be viewed as a technical rectification construction (akin to passing from quasicategories to simplicial categories) that improves strictness properties, but does not change the presented higher object.

\begin{construction}
  We define the opposite of the (strict) category \(\cat{C}\) in \(\Univ_{0, \LexOp}\).
  The objects are, up to equivalence, \(\UKappa\)-presentable \(\bbT\)-models.
  A morphism between \(\bbT\)-models \(A\) and \(B\) is a function \(h \co A \to B\) with, for each \(o \co O_n\), a map
  \[
    o_h \co {\prod_{\substack{ u \co A^n,\, v \co A }}} o_A(u) \simeq_A v \longrightarrow o_B\paren[\big]{ h^n(u) } \simeq_B h(v).
  \]
  Identities are given by identities in each component.
  The composition of morphisms \(g\) and \(h\) is given by the composition of the underlying functions and \(o_{h \cc g}(u, v) \coloneqq o_h\paren[\big]{ g(u), g(v) } \cc o_g(u, v)\).
\end{construction}
\begin{proof}
  By \cref{rem:category-of-models-size}, the universe \(\Univ_{0,\LexOp}\) has a type classifying, up to equivalence, \(\UKappa\)-presentable \(\bbT\)-models.
  Functions compose strictly.
  Thus, composition is strictly unital and associative.
\end{proof}

\begin{lemma}\label{prop:duality:strictification}
  We have an equivalence of categories \(\cat{C}^{\op} \simeq \Mod*[\bbT,\kappa]\) (in the sense of \(\HoTT\)).
\end{lemma}
\begin{proof}
  By contracting the singleton formed by \(v\) and \(o_A(u) \simeq_A v\).
\end{proof}

By construction, the category \(\cat{C}\) has a type of objects and family of morphisms that are fibrant and \(\LexOp\)-modal.
For the rest of the paper, we instantiate \cref{ass:cat,ass:fibrant-cat} with \(\cat{C}\).

\subsubsection{Algebras}
For stating and verifying the axioms, we will need the following notion.

\begin{definition}\label[definition]{def:duality:algebra}
  Given a \(\bbT\)-model \(A\), an \emph{\(A\)-algebra} is a \(\bbT\)-model \(B\) with a morphism \(A \to B\).
\end{definition}

Note that \(A\)-algebras are again models of an algebraic theory \(\bbT/A\), obtained by adjoining new constants \(c_a\) for each \(a \co A\) and equations forcing that \(h(a) \coloneqq c_a\) defines a morphism from \(A\).
This yields a notion of \(\UKappa\)-presentable \(A\)-algebra.
Under this identification, the coprojection \(A \to \sF_{\bbT} X \otimes_{\bbT} A\) corresponds to the free \(\bbT/A\)-model on \(X\).

\begin{lemma}\label[lemma]{prop:duality:combine-presentable-models}
  If \(A\) is a \(\UKappa\)-presentable \(\bbT\)-model and \(B\) is a \(\UKappa\)-presentable \(A\)-algebra, then \(B\) is also a \(\UKappa\)-presentable \(\bbT\)-model.
\end{lemma}
\begin{proof}
  By projectivity of types in \(\UKappa\).
\end{proof}

\subsection{The generic model}\label[subsection]{sec:duality:generic-model}

We work internally to \(\cSet\).
We adapt the construction of a classifying topos for an algebraic theory (see, \eg,~\cite[D3.1, Theorem 3.1.1]{Johnstone2002Elephant2}) to our setting: we construct a \(\bbT\)-model in \(\Mod*[\bbT,\kappa]^{\mathsf{op}}\) and then take its image under the Yoneda embedding to obtain a model in presheaves.
Crucially, our base category only has products up to homotopy.

\begin{lemma}\label[lemma]{prop:duality:site-has-weak-products}
  For \(I \co \Univ_{\Fib*}\) and \(A \co I \to \Mod*[\bbT,\kappa]\), the following comparison map is a levelwise equivalence:
  \(
    \pair{\yo \iota_i}_{i: I} \co \yo \paren[\big]{{\textstyle\bigotimes_{i: I}A}} \longrightarrow \prod_{i\co I} \yo A_i.
  \)
\end{lemma}

The notion of \(\bbT\)-model from \cref{sec:duality:models} makes sense in every category \(\cat{C}\) with homogeneously fibrant hom-spaces and weak finite products, meaning that the map \(\yo (A \times B) \to \yo A \times \yo B\) is levelwise an equivalence.\footnote{We only require fibrancy since theories descend in the expected way to lex reflective subuniverses.}
A model is given by an object \(A\) such that each hom-space into \(A\) is an h-set, a morphism \(A^n \to A\) for each \(o \co O_n\), and for each equation a path in the corresponding hom-h-set.
A morphism of \(\bbT\)-models is a morphism such that the evident squares commute up to path in the hom-h-set.

\begin{construction}\label[construction]{prop:duality:generic-model-base}
  The object \(\sF 1 \co \Mod*[\bbT,\kappa]^{\mathsf{op}}\) has canonically the structure of a \(\bbT\)-model.
\end{construction}
\begin{proof}
  For each \(o \co O_n\), we have to give a morphism \(\sF 1 \to \sF n\).
  We pick the morphism associated to the element \(o_{\sF n}(1, \ldots, n) \co \sF n\).
  The required paths exist by the construction of the free model.
\end{proof}

\begin{construction}
  There is a \(\bbT\)-model with underlying object \(\LO{D}\paren{\yo (\sF 1)}\) in \(\PSh*{\cat{C}}\).
  We denote this \(\bbT\)-model by \(\sG_{\bbT}\) (or \(\sG\) if clear from context) and refer to it as the \emph{generic \(\bbT\)-model}.
\end{construction}
\begin{proof}
  By \cref{prop:duality:site-has-weak-products}, the image of \(\sF 1\) under \(\sD \cc \yo\) has canonically the structure of a \(\bbT\)-model.
  It is an h-set by \cref{prop:cobar-strictifies-n-types} since \(\Mod*[\bbT,\kappa]^{\mathsf{op}}(A, \sF 1)\) is an h-set for all \(A\).
\end{proof}

\begin{remark}\label[remark]{rem:duality:generic-model-characterization}
  We define \(\sG_0 \co \PSh*{\cat{C}_0}\) by \(\sG_0(A) \coloneqq A\), which is a family of \(\LO{T}\)-modal types and has pointwise a \(\bbT\)-model structure.
  There is a canonical equivalence \(\sG(A) \simeq A\) given by the component of \(\tau_{\yo (\sF 1)} \co \yo\paren{\sF 1} \to \sG\) at \(A\) and the equivalence \(\Mod*(\sF 1, A) \simeq A\).
  This equivalence preserves the \(\bbT\)-model structure.
  Thus, \(\PM{U}\sG\) and \(\sG_0\) are equivalent \(\bbT\)-models in \(\PSh*{\cat{C}_0}\).
\end{remark}

\subsection{Statement of the axioms}\label[subsection]{sec:duality:axioms}

The theory \(\bbT\) lifts to a theory in \(\PSh{\cat*{C}}\) by taking constant cubical presheaves (note that constant presheaves on a fibrant cubical set are fibrant).
The definition of a \(\bbT\)-model internally to \(\PSh{\cat*{C}}\) agrees with the notion of model in a category with weak finite products from \cref{sec:duality:generic-model}.
In presheaves, we will only consider \(\LO{DT}\)-modal \(\bbT\)-models.
Free \(\bbT\)-models and colimits of \(\bbT\)-models exist in \(\PSh{\cat*{C}}_{\LO{D}\LO{T}}\) by the same reasoning as in \Cref{sec:duality:models}.
By construction, \(\yo (\sF 1)\) is levelwise \(\LexOp\)-modal and thus \(\sG\) is \(\sD\LexOp\)-modal by \ref{prop:lifting-of-lex-operations-over-fibrant-sites:modal-type-characterization}.
\Cref{def:duality:algebra} yields a notion of \(\sG\)-algebra internally to \(\PSh{\cat*{C}}\).

\begin{definition}[Internally to \(\PSh{\cat*{C}}_{\LO{D}\LO{T}}\)]
  Given a \(\sG\)-algebra \(h \co \sG \to A\), its \emph{spectrum} \(\Spec(A)\) is the type \(\Mod*[\bbT/\sG](A, \sG) = \sum_{r\co \Mod*[\bbT](A, \sG)} r \cc h \sim \id_{\sG}\), with associated \emph{evaluation map} \(\ev_A \co \Spec(A) \times A \to \sG\) sending \((r, a)\) to \(r(a)\).
\end{definition}

Similarly, since presheaves constant on a fibrant cubical set are fibrant, the universe \((\UKappa, \El*_{\kappa})\) yields a fibrant, levelwise \(\LexOp\)-modal universe of fibrant, levelwise \(\LexOp\)-modal types in presheaves.
The replacement \((\sD \UKappa, \widetilde{\sD}\El*_{\kappa})\) yields a \(\sD\LexOp\)-modal universe of \(\sD\LexOp\)-modal types.%
\footnote{We do not explain how to lift the closure properties of \(\UKappa\); these are not required for our development.}
For finite and countable types, this universe agrees with the expected one in \(\PSh{\cat*{C}}_{\LO{D}\LO{T}}\).
\footnote{%
  In the countable case, we use that \(\sD\) preserves products (\cf \Cref{prop:cobar-pseudomorphism:preservation-properties}) to show that the comparison map \(\sD\kappa \simeq \sD\paren[\big]{ \paren{\LexOp 2}^{\paren{\LexOp \sN}} } \to \paren{\sD\LexOp 2}^{\paren{\sD\LexOp \sN}}\) is an equivalence.
  Writing \(\varepsilon\) for the evaluation map, the universal countable type over \(\paren{\sD\LexOp 2}^{\paren{\sD\LexOp \sN}}\) is \(\El*' f = \sum_{u \co \sD\LexOp \sN} \varepsilon(f,u) \simeq_{\sD\LexOp 2} 1\).
  Pulling back along the comparison map yields \(\El*'' f = \sum_{u \co \sD\LexOp \sN} (\sD\varepsilon)(f,u) \simeq_{\sD\LexOp 2} 1\), which is equivalent to \(\widetilde{\sD}\El*_{\UKappa}\) by left exactness of \(\sD\).
}

We can now state the axioms we are verifying in the model \(\PSh{\cat*{C}}_{\LO{D}\LO{T}}\).

\begin{axiom}\label[axiom]{axiom:projectivity-of-spec}
  For all \(A \co \Mod*[\bbT/\sG,\LO{D}\kappa]\), the type \(\Spec(A)\) is projective (\cf~\cref{def:projective}).
\end{axiom}
\begin{axiom}\label[axiom]{axiom:duality}
  For all \(A \co \Mod*[\bbT/\sG,\LO{D}\kappa]\), the transposed evaluation \(A \to \sG^{\Spec(A)}\) an equivalence.
\end{axiom}
\begin{axiom}\label[axiom]{axiom:constant-are-spec-null}
  For all \(A \co \Mod*[\bbT/\sG,\LO{D}\kappa]\), the constant map \(\sN \to \sN^{\Spec(A)}\) is an equivalence.
\end{axiom}

In the case of finitely presented \(k\)-algebras, Cherubini, Coquand, and Hutzler~\cite{CherubiniCoquandHutzler2024} construct a model of synthetic algebraic geometry as an inner model, starting from the above axioms, by nullifying a family of propositions.%
\footnote{%
As shown in~\cite{CoquandHoferSattler2025}, \cref{axiom:duality} is not constructively valid in the model of~\cite{CherubiniCoquandHutzler2024}.
We bridge this gap.%
}
Classically, this corresponds to the passage from higher presheaves to higher sheaves.
For synthetic Stone duality~\cite{CherubiniCoquandGeerligsMoeneclaey2024}, this step is similar.

\subsection{Local representability of the spectrum}

To verify \cref{axiom:duality,axiom:constant-are-spec-null}, we will appeal to the fact that equivalences between \(\sD\)-modal types are levelwise by \cref{prop:cobar-strictifies-equivalences-between-levelwise-modal-types}.
This requires us to characterize objects in the presheaf model levelwise.
Furthermore, we need to characterize dependent products over \(\Spec\).
The key result is \cref{prop:duality:spec-locally-representable}, stating that the spectrum of a \(\LO{D}\UKappa\)-presentable algebra is \emph{locally homotopy representable} in the following sense.

\begin{definition}[\cref{def:representables:locally-representable}, internally to \(\cSet\)]
  Let \(S \co \Ty*_{\PSh*{\cat{C}}}^{\Fib*_0}(P)\) be a type.
  A \emph{local homotopy representation} of \(S\) is given, for each \(x \co \cat{C}_0\) and \(\gamma \co \Gamma_x\), by an object \(x.\gamma\) with a morphism \(\p_\gamma \co x.\gamma \to x\), and some \(\q_\gamma \co S(x.\gamma, \gamma\p_\gamma, \q_\gamma)\) together with a witness of a levelwise equivalence for the induced map \(\pair{\p_\gamma, \q_\gamma} \co \yo(x.\gamma) \to \yo x.S\gamma\).
\end{definition}

The above is a version of a (locally) representable natural transformation in our setting.
These were used by Awodey~\cite{Awodey2018} to define natural models of type theory.
They play a similar r\^{o}le in our construction.%
\begin{remark}
  The local representability of \(\Spec\) over \(\Pr\) is analogous to constructing a natural model on the base category.
  \Cref{prop:duality:base-category-is-cwa} can be understood as the opposite of the base category being a CwF/CwA.
  The families \(\Pr_{\bbT/\sG_0,\kappa}\) and \(\Ret\) defined in \cref{sec:cwf-structure-on-base-category,sec:characterization-of-spec}, which characterize the types \(\Pr\) and \(\Spec\) levelwise, are the analogues of the presheaves of types and elements.
  \Cref{prop:duality:spec-locally-representable} can then be seen as one direction of the equivalence between CwFs/CwAs and natural models.%
\footnote{This part of the construction of a model of duality can therefore be seen as a special case of the construction of the standard model of two-level type theory~\cite{AnnenkovCapriottiKrausSattler2023} where the presheaves of inner types and elements are definable internally to presheaves in terms of the generic model \(\sG\).}
\end{remark}

\subsubsection{Families of presentations and algebras}\label[subsubsection]{sec:cwf-structure-on-base-category}

Internally to \(\CubicalSet\), we construct \(\Pr_{\bbT/\sG_0, \kappa}\) over \(\Mod*[\bbT,\kappa]\) as \(\Pr_{\bbT/\sG_0, \kappa}(A) \coloneqq \Pr_{\bbT/A, \kappa}\).
This is functorial up to homotopy: for \(h \co A \to B\) and \(u = (X, R, s, t) \co \Pr_{\bbT/A, \kappa}\), we have \(uh \co \Pr_{\bbT/B, \kappa}\) given by composing \(s, t \co R \rightrightarrows \sF_{\bbT/A}(X)\) with the map \(\sF_{\bbT/A}X \to \sF_{\bbT/B}X\) induced by \(h\).

We furthermore construct a family of \(\bbT\)-models (and also \(\sG_0\)-algebras) over \(\Pr_{\bbT/\sG_0, \kappa}\).
This family is at \(A\) and \(u\) given by \(\gen{u}_{\bbT/A}\), and comes with morphism \(\p_u \co A \to \gen{u}_{\bbT/A}\) of \(\bbT\)-models.
The following lemma shows in particular that this family is also functorial up to homotopy.

\begin{lemma}[Internally to \(\cSet\)]\label[lemma]{prop:duality:base-category-is-cwa}
  For \(A, B \co \Mod*[\bbT,\kappa]\), \(h \co A \to B\), and \(u \co \Pr_{\bbT/A, \kappa}\), there is some \(h^+ \co \gen{u}_{\bbT/A} \to \gen{uh}_{\bbT/B}\) making the following into a pushout up to homotopy in \(\Mod*[\bbT,\kappa]\):
  \begin{center}\(\begin{tikzcd}
    A
    \ar[r, "\p_u"]
    \ar[d, "h"']
  &
    {\gen{u}_{\bbT/A}}
    \ar[d, "h^+"]
  \\
    B
    \ar[r, "\p_{uh}"']
  &
    {\gen{uh}_{\bbT/B}\rlap{.}}
    \ar[ul, pushout]
  \end{tikzcd}\)\end{center}
  Furthermore, this assignment satisfies \(\id^+ \sim \id\) and \((fg)^+ \sim f^+g^+\).
\end{lemma}
\begin{proof}
  Denote the presentation by \(u = (X, R, s, t)\).
  The \(\bbT\)-models \(\gen{u}\) and \(\gen{uh}\) are in \(\Mod*[\bbT,\kappa]\) by \cref{prop:duality:combine-presentable-models}.
  The above square arises as the outer square in the following diagram:
  \begin{center}\(\begin{tikzcd}[column sep=2.0cm]
    A
    \ar[r, "\iota_2"]
    \ar[d, "h"']
  &
    {\sF R \otimes A}
    \ar[r, shift left=1, "{[s, \iota_2]}"]
    \ar[r, shift right=1, "{[t, \iota_2]}"']
    \ar[d, "\sF R \otimes h"']
  &
    {\sF X \otimes A}
    \ar[d, "\sF X \otimes h"]
    \ar[r, two heads, "q_u"]
  &
    {\gen{u}}
    \ar[d, "h^+", dashed]
  \\
    B
    \ar[r, "\iota_2"']
  &
    {\sF R \otimes B}
    \ar[r, shift left=1, "{[(\sF X \otimes h)s, \iota_2]}"]
    \ar[r, shift right=1, "{[(\sF X \otimes h)t, \iota_2]}"']
  &
    {\sF X \otimes B}
    \ar[r, two heads, "q_{uh}"']
  &
    {\gen{uh}\rlap{.}}
  \end{tikzcd}\)\end{center}
  The parallel center squares commute by construction and therefore \(h^+\) exists by the universal property of the top coequalizer.
  The left square is a pushout by construction.
  The two center squares are pushouts by pasting with the left square.

  Suppose we are given \(f \co \gen{u} \to C\) and \(g \co \sF X \otimes B \to C\) with \(f \cc q_u \sim g \cc (\sF X \otimes h)\).
  It suffices to find some \(\overline{g} \co \gen{uh} \to C\) with \(\overline{g} \cc q_{uh} \sim g\) since \(q_u\) and \(q_{uh}\) are epi.
  We have
  \[
    g \cc (\sF X \otimes h)\cc s \sim f \cc q_u \cc s \sim f \cc q_u \cc t \sim g \cc (\sF X \otimes h)\cc t \rlap{,}
  \]
  so \(g\) makes the bottom parallel maps equal.
  The universal property of \(q_{uh}\) gives the claim.

  The equations \(\id^+ \sim \id\) and \((fg)^+ \sim f^+ g^+\) follow from the uniqueness of the morphism out of the defining coequalizer.
\end{proof}

\subsubsection{Levelwise characterizations of \texorpdfstring{\(\bbT\)}{𝕋}-models}

To characterize the objects in the statement of the axioms levelwise, we will lift \(\PM{U}\) and the right adjoint on types \(\LO{R}\) from \cref{prop:cofree-presheaf:sends-levelwise-modal-types-for-base-model-modality-to-modal-types-for-combined-modality} further to \(\bbT\)-models.
For the model of \(\LO{D}\LO{T}\)-modal types \(\PSh{\cat*{C}}_{\LO{D}\LO{T}}\), we write \(\Mod[\PSh{\cat*{C}}_{\LO{D}\LO{T}}](\Gamma)\) for the set of \(\bbT\)-models in context \(\Gamma \in \PSh{\cat*{C}}_{\LO{D}\LO{T}}\) and \(\Mod*[\PSh{\cat*{C}}_{\LO{D}\LO{T}}](\Gamma)(A, B)\) for the cubical set of homomorphism between models \(A\) and \(B\).
We use analogous notation for the model \(\cSet_{\LO{T}} \comma \cat{C}_0\).

Internally to \(\CubicalSet\), both \(\PM{U}\) and \(\sR\) preserve products by \cref{prop:forgetful:preservation-properties,prop:cofree-presheaf:preservation-properties}.
Therefore, they canonically act on \(\bbT\)-models in the models of \(\LO{D}\LO{T}\)-modal and \(\LO{T}\)-modal types respectively.
Furthermore, the unit and counit of the adjunction lift to strict homomorphisms, meaning that all required squares commute strictly.
Hence, we have the following.

\begin{construction}\label[construction]{prop:duality:adjunction-on-models}
  Naturally in \(\Gamma \in \PSh{\cat*{C}}\) and naturally up to homotopy in \(A \in \Mod[\PSh{\cat*{C}}_{\LO{D}\LO{T}}](\Gamma)\) and \(B \in \Mod[\cSet_{\LO{T}} \comma \cat{C}_0](\PM{U}\Gamma)\), there is an isomorphism over the isomorphism from \cref{prop:cofree-presheaf:dependent-right-adjoint}:
  \(
    \Mod*[\cSet_{\LO{T}} \comma \cat{C}_0](\PM{U}\Gamma)(\PM{U} A, B) \cong \Mod*[\PSh{\cat*{C}}_{{\LO{D}\LO{T}}}](\Gamma)(A, \LO{R} B).
  \)
\end{construction}

In the model of type theory \(\cSet_{\Fib} \comma \cat{C}_0\), all constructions on \(\bbT\)-models are performed levelwise.
By the preservation of \(\bbT\)-models and morphisms by \(\PM{U}\), we can construct in particular substitution-stable comparison maps for colimit-like constructions on \(\bbT\)-models: \(\bigotimes_{\PM{U} I} \PM{U} A \to \PM{U}(\bigotimes_{I} A)\), \(\sF (\PM{U} X) \to \PM{U}(\sF X)\), etc.
We show that \(\PM{U}\) preserves these constructions in the sense that all these maps can be equipped with the structure of an equivalence in a substitution-stable way.

\begin{lemma}\label[lemma]{prop:dualtiy:levelwise-chararcterization-of-t-model-constructions}
  The pseudomorphism \(\PM{U} \co \PSh{\cat*{C}}_{\sD\LexOp} \to \cSet_{\LexOp} \comma \cat{C}_0\) preserves free \(\bbT\)-models, coproducts of \(\bbT\)-models, and coequalizers of \(\bbT\)-models up to equivalence.
\end{lemma}
\begin{proof}
  Working (strictly) naturally in \(\Gamma \in \PSh{\cat*{C}}\), by an adjoint argument up to homotopy using \cref{prop:duality:adjunction-on-models}.
\end{proof}

\begin{corollary}\label[corollary]{prop:dualtiy:levelwise-chararcterization-of-g-algebra-constructions}
  \(\PM{U} \co \PSh{\cat*{C}}_{\sD\LexOp} \to \cSet_{\LexOp} \comma \cat{C}_0\) sends up to equivalence \(\sG\)-algebras to \(\sG_0\)-algebras, and sends free (\resp coproducts of, coequalizers of) \(\sG\)-algebras to free (\resp coproducts of, coequalizers of) \(\sG_0\)-algebras.
\end{corollary}
\begin{proof}
  By \cref{rem:duality:generic-model-characterization}, \(\PM{U}\sG \simeq \sG_0\).
  Thus, a model under \(\sG\) is sent to a model under \(\sG_0\).
  All constructions mentioned above reduce to constructions on \(\bbT\)-models.
\end{proof}

\subsubsection{Levelwise characterization of the spectrum}\label{sec:characterization-of-spec}

In the statements of \cref{axiom:duality,axiom:projectivity-of-spec,axiom:constant-are-spec-null}, we only consider \(\sD\UKappa\)-presentable algebras.
These axioms are logically equivalent to versions assuming a given \(\sD\UKappa\)-presentation.
The type of \(\sD\UKappa\)-presentations of \(\sG\)-algebras in \(\PSh{\cat*{C}}_{\sD\LexOp}\) is denoted \(\Pr_{\bbT/\sG,\sD \kappa}\).
Because of the lifting via constant presheaves, we have that \(\PM{U}\) preserves the universe \(\UKappa\), and furthermore the following.

\begin{lemma}\label[lemma]{prop:duality:ukappa-exponentials}
  The pseudomorphism \(\PM{U} \co \PSh{\cat*{C}}_{\LO{D}\LO{T}} \to \cSet_{\LO{T}} \comma \cat{C}_0\) preserves dependent products with domain in \((\UKappa, \El*_\kappa)\).
\end{lemma}
\begin{proof}
  Internally to \(\cSet\), this corresponds to the fact that dependent products with fiberwise constant domain are computed levelwise.
\end{proof}

\begin{construction}[Internally to \(\cSet\)]\label[construction]{prop:duality:presentation-characterization}
  Naturally up to homotopy in \(A \co \Mod*[\bbT,\kappa]\), there is an equivalence \(\Pr_{\bbT/\sG, \sD\kappa}(A) \simeq \Pr_{\bbT/A,\kappa}\).
\end{construction}
\begin{proof}
  We construct an equivalence in three steps.
  First, we construct an equivalence in \(\PSh{\cat*{C}}\).
  The unit \(\tau_{\UKappa} \co \UKappa \to \sD \UKappa\) and universal property of the \(\sD\)-modal replacement yield the following comparison map that is a levelwise equivalence:
  \begin{align}\label{eq:duality:presentation-characterization:comparison}
    \Pr_{\bbT/\sG,\sD\kappa}
    \leftarrow \sum_{X, R \co \UKappa} \paren[\Big]{\sF_{\bbT/\sG}\paren[\big]{\sD (\El*_\kappa X)}^{\sD (\El*_\kappa R)} }^2
    \simeq \sum_{X, R \co \UKappa} \paren[\Big]{\sF_{\bbT/\sG}\paren[\big]{\sD (\El*_\kappa X)}^{\El*_\kappa R} }^2.
  \end{align}
  Applying \(\PM{U}\) to the above yields an equivalence.
  We denote the last type above by \(\Pr_{\bbT/\sG,\sD \kappa}'\).
  Second, by \cref{prop:dualtiy:levelwise-chararcterization-of-g-algebra-constructions,prop:duality:ukappa-exponentials}, we have the first equivalence in \(\cSet_{\LO{T}} \comma \cat{C}_0\) below:
  \begin{equation}\label{eq:duality:presentation-characterization:step-2-and-3}
    \PM{U}\Pr_{\bbT/\sG,\sD \kappa}'
    \simeq \sum_{X, R \co \kappa} \paren[\Big]{ \paren[\big]{ \sF_{\bbT/\sG_0}( \PM{U}( \LO{D}\El*_\kappa )(X) ) }^{\El*_\kappa R} }^2
    \simeq \sum_{X, R \co \kappa} \paren[\Big]{ \paren[\big]{ \sF_{\bbT/\sG_0}( \El*_\kappa X ) }^{\El*_\kappa R} }^2.
  \end{equation}
  Third, we use that \(\PM{U}\tau_{\El*_\kappa} \co \El*_\kappa \to \sD \El*_\kappa\) in context \(\kappa\) is an equivalence by \cref{prop:cobar-unit-is-levelwise-equivalence}, yielding the second equivalence above.
\end{proof}

We now characterize the universal \(\LO{D}\UKappa\)-presented \(\sG\)-algebra over \(\Pr_{\bbT/\sG,\sD \kappa}\) levelwise in terms of the family of presented algebras from \cref{sec:cwf-structure-on-base-category}.

\begin{construction}[Internally to \(\cSet\)]\label[construction]{prop:duality:presented-algebra-characterization}
  Naturally up to homotopy in \(A \co \Mod*[\bbT,\kappa]\), over the equivalence from \cref{prop:duality:presentation-characterization}, there is an equivalence \(\sum_{\Pr_{\bbT/\sG,\sD\kappa}(A)} \gen{-}_{\bbT/\sG}(A) \simeq \sum_{\Pr_{\bbT/A,\kappa}} \gen{-}_{\bbT/A}\) that is fiberwise an equivalence of the presented \(\bbT\)-models.
\end{construction}
\begin{proof}
  We construct an equivalence over the equivalence from \cref{prop:duality:presentation-characterization} in three steps.
  The family \(\gen{-}_{\bbT/\sG}\) is at \((X, R, s, r) \co \Pr_{\bbT/\sG, \LO{D}\kappa}\) given by the coequalizer of the homomorphisms \(\overline{s}, \overline{r} \co \sF_{\bbT/\sG}\paren[\big]{(\widetilde{\sD}\El*_\kappa) (R)} \rightrightarrows \sF_{\bbT/\sG}\paren[\big]{(\widetilde{\sD}\El*_\kappa) (X)}\) which are the transposes of the maps \(s, r \co (\widetilde{\sD}\El*_\kappa)(R) \rightrightarrows \sF_{\bbT/\sG}\paren[\big]{(\widetilde{\sD}\El*_\kappa) (X)}\) under the free-forgetful adjunction for \(\bbT/\sG\)-models.

  In the first step, we pull back the family along the comparison map~\eqref{eq:duality:presentation-characterization:comparison}.
  We obtain the family \(\gen{-}_{\bbT/\sG}'\) that at \((X, R, s, r) \co \Pr_{\bbT/\sG,\LO{D}\kappa}'\) (where \(\Pr_{\bbT/\sG,\LO{D}\kappa}'\) is defined as in \cref{prop:duality:presentation-characterization}) is given by the coequalizer of \(\overline{p_{X,R}(s)}, \overline{p_{X,R}(r)} \co \sF_{\bbT/\sG}\paren[\big]{(\widetilde{\sD}\El*_\kappa) (R)} \rightrightarrows \sF_{\bbT/\sG}\paren[\big]{(\widetilde{\sD}\El*_\kappa) (X)}\).
  Here, \(\overline{\,\cdot\,}\) denotes again the transpose and \(p_{X,R}\) denotes the inverse to restriction along \(\tau_{\El*_\kappa R}\) from~\eqref{eq:duality:presentation-characterization:comparison}.

  In the second step, we construct an equivalence in \(\cSet_{\LO{T}} \comma \cat{C}_0\) from \(\PM{U}\gen{-}_{\bbT/\sG}'\) over the first equivalence from~\eqref{eq:duality:presentation-characterization:step-2-and-3}.
  We have that \(\PM{U}\) preserves coequalizers and free \(\bbT/\sG\)-models by \cref{prop:dualtiy:levelwise-chararcterization-of-g-algebra-constructions} and exponentials with domain in \(\UKappa\) by \cref{prop:duality:ukappa-exponentials}.
  Since restriction along \(\tau_{\El*_\UKappa R}\) is given componentwise by restriction along the components of \(\tau_{\El*_\UKappa R}\), componentwise, the action of \(p\) is homotopic to restriction along the components of the inverse from \cref{prop:cobar-unit-is-levelwise-equivalence}. 
  We obtain an equivalence to the family given at \(A \co \Mod*[\bbT,\kappa]\), \(X, R \co \kappa\), and \(s, r \co \El*_\kappa R \rightrightarrows \sF_{\bbT/A}\paren[\big]{ (\LO{D} \El*_\kappa)(A, X) }\) by the coequalizer of the homomorphisms \(\overline{s \cc \tau^{-1}}, \overline{r \cc \tau^{-1}} \co \sF_{\bbT/A}\paren[\big]{ (\LO{D} \El*_\kappa)(A, R) } \rightrightarrows \sF_{\bbT/A}\paren[\big]{ (\LO{D} \El*_\kappa)(A, X) }\).
  Here, \(\tau^{-1}\) denotes the corresponding component of the inverse from \Cref{prop:cobar-unit-is-levelwise-equivalence}.

  In the third step, we construct an equivalence with the previous family over the second equivalence from~\eqref{eq:duality:presentation-characterization:step-2-and-3} using that all components of \(\tau\) are equivalences by \cref{prop:cobar-unit-is-levelwise-equivalence} and that all constructions on \(\bbT/A\)-models respect equivalences.
\end{proof}

Internally to \(\CubicalSet\), we construct a family \(\Ret\) over \(\Pr_{\bbT/\sG_0, \kappa}\).
At \(A \co \Mod*_{\bbT,\UKappa}\) and \(u \co \Pr_{\bbT/A, \kappa}\), we define \(\Ret_A(u)\) as the type of retractions of \(\p_u \co A \to \gen{u}_{\bbT/A}\).
We have an evaluation map \(\ev_{A,u} \co \Ret_A(u) \times \gen{u}_{\bbT/A} \to A\) given by evaluation of the retraction.

\begin{construction}[Internally to \(\cSet\)]\label[construction]{prop:duality:spec-characterization}
  Naturally up to homotopy in \(A \co \Mod*[\bbT,\kappa]\), there is an equivalence \(\sum_{\Pr_{\bbT/\sG,\sD\kappa}(A)} \Spec(A) \simeq \sum_{\Pr_{\bbT/A,\kappa}} \Ret_A\) over the equivalence from \cref{prop:duality:presentation-characterization}.
  Under this equivalence and the ones of \cref{rem:duality:generic-model-characterization,prop:duality:presented-algebra-characterization}, the evaluation map of the spectrum corresponds to the evaluation map for $\Ret_A$.
\end{construction}

\begin{proof}
  Let \(u = (X, R, s, t) \co \Pr_{\bbT/\sG,\sD \kappa}\), \ie, \(X, R \co \sD\UKappa\) and \(s, t \co (\widetilde\sD \El*_\kappa) (R) \rightrightarrows \sF_{\bbT/\sG}\paren[\big]{ (\widetilde\sD \El*_\kappa) (X) }\).
  Denoting the evaluation map of the free \(\bbT/\sG\)-algebra by \(\varepsilon\), the spectrum satisfies
  \[
    \Spec(\gen{u}) \simeq \sum_{ v \co \sG^{(\tilde{\sD}\El*_\kappa) (X)} } \prod_{ r \co (\widetilde\sD \El*_\kappa) (R) } \paren[\big]{ \varepsilon_{v}(s_r) \simeq_{\sG} \varepsilon_{v}(t_r) }.
  \]
  Now the argument is analogous to the one from \cref{prop:duality:presented-algebra-characterization}.
  First, we construct an equivalence in \(\PSh{\cat*{C}}\).
  Pulling back the above family over \(\Pr_{\bbT/\sG,\sD\kappa}\) along the comparison map from \(\Pr_{\bbT/\sG,\sD\kappa}'\) from \eqref{eq:duality:presentation-characterization:comparison}, we obtain the family that at \((X, R, s, t) \co \Pr'_{\bbT/\sG, \sD\kappa}\) is given by
  \begin{multline*}
    \sum_{ v \co \sG^{\sD(\El*_\kappa X)} } \prod_{ r \co {\sD(\El*_\kappa R)} } \paren[\Big]{ \varepsilon_{v}\paren[\big]{(\tau^*_{\El*_\kappa R})^{-1}(s)_r} \simeq_{\sG} \varepsilon_{v}\paren[\big]{ (\tau^*_{\El*_\kappa R})^{-1}(s)_r } }
    \\ \simeq
    \sum_{ v \co \sG^{\El*_\kappa X} } \prod_{ r \co {\El*_\kappa R} } \paren[\big]{ \varepsilon_{(\tau^*_{\sG})^{-1}(v)}(s_r) \simeq_{\sG} \varepsilon_{(\tau^*_{\sG})^{-1}(v)}(t_r) }.
  \end{multline*}
  We denote the latter by \(\Spec'\).
  In the second step, we construct an equivalence in \(\cSet_{\LO{T}} \comma \cat{C}_0\) from \(\PM{U}\Spec'\) over the first equivalence from~\eqref{eq:duality:presentation-characterization:step-2-and-3}.
  Analogous to \cref{prop:duality:presented-algebra-characterization}, using additionally that \(\PM{U}\) preserves paths by \cref{prop:forgetful:preservation-properties}, we obtain an equivalence to the family given at \(A \co \Mod*[\bbT,\kappa]\), \(X, R \co \kappa\), and \(s, t \co \El*_\kappa R \rightrightarrows \sF_{\bbT/A}\paren{ (\LO{D} \El*_\kappa)(A, X) }\) by
  \[
    \sum_{ v \co A^{\El*_\kappa X} } \prod_{r \co {\El*_\kappa R}} \paren[\big]{ \varepsilon_{v \cc \tau^{-1}}(s_r) \simeq_A \varepsilon_{v \cc \tau^{-1}}(t_r) }.
  \]
  Third, we construct an equivalence from the previous family, over the second equivalence from~\eqref{eq:duality:presentation-characterization:step-2-and-3}, to the family given at \(A \co \Mod*[\bbT,\kappa]\) and \(u = (X, R, s, t) \co \Pr_{\bbT, \UKappa}\), where \(X, R \co \kappa\) and \(s, t \co \El*_\kappa R \rightrightarrows \sF_{\bbT/A}\paren{ \El*_\kappa X }\), by
  \[
    \sum_{ v \co A^{\El*_\kappa X} } \prod_{r \co {\El*_\kappa R}} \paren[\big]{ \varepsilon_{v}(s_r) \simeq_A \varepsilon_{v}(t_r) }.
  \]
  This family is equivalent to the family of retractions of \(\p \co A \to \gen{u}_{\bbT/A}\).

  The claim about the evaluation map follows by tracing through the above steps simultaneously with those of \cref{rem:duality:generic-model-characterization,prop:duality:presented-algebra-characterization}.
\end{proof}

Internally to \(\cSet\), given \(A \co \Mod*[\bbT,\kappa]\), recall the morphism \(\p_u \co A \to \gen{u}_{\bbT/A}\) of \(\bbT\)-models for \(u \co \Pr_{\bbT/A, \kappa}\).
In \(\bbT\)-models, the identity on \(\gen{u}_{\bbT/A}\) induces a common retraction \(\q_u \co \gen{u\p_u}_{\bbT/A} \to \gen{u}_{\bbT/A}\) of the maps \(\p_{u\p_u}, \p_u^+ \co \gen{u}_{\bbT/A} \to \gen{u\p_u}_{\bbT/A}\).
Hence, \(\q_u \co \Ret_u(u\p_u)\).

\begin{proposition}\label[proposition]{prop:duality:spec-locally-representable}
  The type \(\Spec \in \Ty_{\PSh{\cat*{C}}}(1.\Pr_{\bbT/\sG,\sD\kappa})\) has the structure of a locally homotopy representable type.
  The representation is at \(A \co \Mod*[\bbT,\kappa]\) and \(u \co \Pr_{\bbT/\sG,\sD \kappa}(A)\), under the equivalence from \cref{prop:duality:spec-characterization}, given by
  \(
    \pair{\p_u, \q_u} \co \yo \gen{u}_{\bbT/A} \longrightarrow \yo A . \Spec(u). 
  \)
\end{proposition}
\begin{proof}
  Let \(A \co \Mod*[\bbT,\kappa]\) and \(u \co \Pr_{\bbT/A, \kappa}\).
  The \(\bbT\)-model \(\gen{u}_{\bbT/A}\) is an object of \(\Mod*[\bbT,\kappa]\) by \cref{prop:duality:combine-presentable-models}.
  For \(B \co \Mod*[\bbT,\kappa]\), we have by singleton contraction and \cref{prop:duality:base-category-is-cwa} the equivalence
  \[
    \Mod*[\bbT,\kappa](\gen{u}, B)
    \simeq \sum_{\substack{ h \co \Mod*[\bbT,\kappa](A, B) \\ r \co \Mod*[\bbT,\kappa](\gen{u}, B) }} \paren{r \cc \p_{u} \sim \id_{B} \cc h}
    \simeq \sum_{ \substack{ h \co \Mod*[\bbT,\kappa](A, B) \\ r \co \Mod*[\bbT,\kappa](\gen{uh}, B) } } \paren{r \cc \p_{uh} \sim \id_{B}}.
  \]
  The underlying map is the claimed one since \(\q_u\) is defined in terms of the universal property of the pushout from \cref{prop:duality:base-category-is-cwa}.
\end{proof}

\subsection{Verification of the axioms}\label[subsection]{sec:duality:axiom-verification}

\Cref{axiom:projectivity-of-spec} follows from the fact that every locally homotopy representable type is projective~(cf.~\cref{prop:lifting-of-lex-operations-over-fibrant-sites:locally-homotopy-representable-types-are-projective}).
For \cref{axiom:duality,axiom:constant-are-spec-null}, we use a characterization of dependent products over a locally representable type.
\begin{lemma}[\cref{prop:pushforward-along-homotopy-locally-representable-transformation}, internally to \(\cSet\)]\label{prop:pushforward-along-homotopy-locally-representable-transformation:restate}
  Naturally in \(\Gamma\), given a locally homotopy representable type \(A \co \Ty*^{\Fib*_0}(\Gamma)\) and a \(\sD\)-modal family \(P \co \Ty*^{\sD}(\Gamma.A)\), the map
  \(
    (\Pi_AP)(x, \gamma) \longrightarrow P(x.\gamma, \gamma\p_\gamma, \q_{\gamma}),
    u \longmapsto u_{\p_{\gamma}}(\q_\gamma)
  \)
  is an equivalence for \(x \co \cat{C}\) and \(\gamma \co \Gamma_x\).
\end{lemma}

\begin{remark}\label[remark]{rem:map-into-pushforward-along-homotopy-locally-representable-transformation}
  Below, we apply \cref{prop:pushforward-along-homotopy-locally-representable-transformation:restate} as follows.
  Let \(A\) and \(P\) be as in \cref{prop:pushforward-along-homotopy-locally-representable-transformation:restate}, \(S \co \Ty*^{\Fib*}(\Gamma)\), and \(f \co \PSh*{\elems\Gamma.A}(S\p_A, P)\).
  The universal property of dependent products induces \(\lambda f \co \PSh*{\elems\Gamma}(S, \Pi_A P)\).
  For \(x \co \cat{C}_0\) and \(\gamma \co \Gamma_x\) the following square commutes strictly:
  \begin{center}\(\begin{tikzcd}[column sep=4em]
    {S(x, \gamma)}
    \ar[d, "S(\p_\gamma)"']
    \ar[r, "{(\lambda f)_{(x, \gamma)}}"]
  &
    {\paren{\Pi_A P}(x, \gamma)}
    \ar[d, "\text{\ref{prop:pushforward-along-homotopy-locally-representable-transformation:restate}}", weq]
  \\
    {S(x.\gamma, \gamma\p_{\gamma})}
    \ar[r, "f_{(x.\gamma, \gamma\p, \q)}"]
  &
    {P(x.\gamma, \gamma\p_{\gamma}, \q_{\gamma})\rlap{.}}
  \end{tikzcd}\)\end{center}
  Hence, by \cref{prop:pushforward-along-homotopy-locally-representable-transformation:restate} and \(2\)-out-of-\(3\), the component of \(\lambda f\) at \((x, \gamma)\) is an equivalence exactly if the composite \(f_{(x.\gamma, \gamma\p, \q)} \cc S(\p_\gamma)\) is one.
\end{remark}

\begin{theorem}[Internally to \(\PSh{\cat*{C}}_{\sD\LexOp}\)]
  The constant map \(\sN \longrightarrow \sN^{\Spec(u)}\)  is an equivalence for all \(u \co \Pr_{\bbT/\sG,\sD \kappa}\).
  Hence, \cref{axiom:constant-are-spec-null} holds.
\end{theorem}
\begin{proof}
  Note that the type of natural numbers in the model of \(\sD\LexOp\)-modal types is given up to equivalence by \(\sD(\Delta \sN)\), the \(\sD\)-modal replacement of the constant presheaf on the \(\LexOp\)-modal natural numbers.
  We denote it by \(\sN_{\sD\LexOp}\).

  The map in question is the transpose of \(\pi_1 \co \sN \times \Spec(u) \to \sN\).
  By externalization of \cref{prop:cobar-strictifies-equivalences-between-levelwise-modal-types}, it suffices to construct an element of \(\El_{\cSet}^{\LexOp}\paren[\big]{ 1.\cat{C}_0.\PM{U}\Pr_{\bbT/\sG,\sD \kappa}, \IsEquiv\paren[\big]{ \PM{U}(\lambda{\pi_1}) } }\).
  We work internally to \(\cSet\).
  Let \(A \co \Mod*[\bbT,\kappa]\) and \(u \co \Pr_{\bbT/A, \kappa}\), and consider the following diagram:
  \begin{center}\(\begin{tikzcd}[column sep=4em]
    {\sN}
    \ar[d, equals]
    \ar[r, weq, "(\tau_{\Delta \sN})_{A}"]
  &
    {\paren{ \sN_{\sD\LexOp} }(A)}
    \ar[r, "(\lambda\pi_1)_{(A, u)}"]
    \ar[d, "\sN_{\sD\LexOp}(\p_u)"]
  &
    {\paren{ \sN_{\sD\LexOp}^{\Spec} }(A, u)}
    \ar[d, weq, "\ref{prop:pushforward-along-homotopy-locally-representable-transformation:restate}"]
  \\
    {\sN}
    \ar[r, "(\tau_{\Delta \sN})_{\gen{u}}", weq]
  &
    {\paren{ \sN_{\sD\LexOp} }(\gen{u})}
    \ar[r, equals]
  &
    {\paren{ \sN_{\sD\LexOp} }(\gen{u})\rlap{.}}
  \end{tikzcd}\)\end{center}
  The left square commutes as \(\tau_{\Delta \sN}\) is natural.
  The right square commutes by \cref{rem:map-into-pushforward-along-homotopy-locally-representable-transformation}.
  By \cref{prop:cobar-unit-is-levelwise-equivalence}, the components of \(\tau_{\Delta \sN}\) are equivalences.
  The claim follows by \(2\)-out-of-\(3\).
\end{proof}

\begin{theorem}[Internally to \(\PSh{\cat*{C}}_{\sD\LexOp}\)]\label[theorem]{prop:duality:duality-axiom}
  The evaluation map \(\gen{u}_{\bbT/\sG} \longrightarrow \sG^{\Spec(\gen{u}_{\bbT/\sG})}\) is an equivalence for all \(u \co \Pr_{\bbT/\sG,\sD \kappa}\).
  Hence, \cref{axiom:duality} holds.
\end{theorem}
\begin{proof}
  By externalization of \cref{prop:cobar-strictifies-equivalences-between-levelwise-modal-types}, it suffices to check that this levelwise, \ie, to construct an element of \(\El_{\cSet}^{\LexOp}\paren[\big]{ 1.\cat{C}_0.\PM{U}\Pr_{\bbT/\sG,\sD \kappa}, \IsEquiv\paren[\big]{ \PM{U}(\lambda{\ev}) } }\).
  We work internally to \(\cSet\).
  Given \(A \co \Mod*[\bbT,\kappa]\) and \(u \co \Pr_{\bbT/\sG,\sD \kappa}(A)\), \cref{rem:map-into-pushforward-along-homotopy-locally-representable-transformation} provides the square
  \begin{center}\(\begin{tikzcd}[column sep=4em]
    {\gen{-}_{\bbT/\sG}(A, u)}
    \ar[d, "\gen{-}_{\bbT/\sG}(\p_{u})"']
    \ar[r, "(\lambda \ev)_{(A, u)}"]
  &
    {\paren{\sG^{\Spec}}(A, u)}
    \ar[d, weq, "\ref{prop:pushforward-along-homotopy-locally-representable-transformation:restate}"]
  \\
    {\gen{-}_{\bbT/\sG}(\gen{u}, u\p)}
    \ar[r, "\ev_{(\gen{u}, u\p, \q)}"]
  &
    {\sG(\gen{u})\rlap{.}}
  \end{tikzcd}\)\end{center}
  We have to show that the top map is an equivalence.
  By \(2\)-out-of-\(3\), it suffices to show that the left-then-bottom composite is an equivalence.
  This is the bottom row in the diagram
  \begin{center}\(\begin{tikzcd}[column sep=5em]
    {\gen{u}}
    \ar[r, "\p_{u}^+"]
    \ar[d, weq', "\ref{prop:duality:presented-algebra-characterization}"']
  &
    {\gen{u\p_{u}}}
    \ar[r, "\q_u"]
    \ar[d, weq, "\ref{prop:duality:presented-algebra-characterization}"]
  &
    {\gen{u}}
    \ar[d, weq, "\ref{rem:duality:generic-model-characterization}"]
  \\
    {\gen{-}_{\bbT/\sG}(A, u)}
    \ar[r, "\gen{-}_{\bbT/\sG}(\p_{u})"]
  &
    {\gen{-}_{\bbT/\sG}(\gen{u}, u\p)}
    \ar[r, "\ev_{(\gen{u}, u\p, \q)}"]
  &
    {\sG(\gen{u})\rlap{.}}
  \end{tikzcd}\)\end{center}
  The left square commutes up to homotopy by the construction of the comparison map.
  The right square commutes by the claim from \cref{prop:duality:spec-characterization} that the evaluation map of the spectrum corresponds to the one for \(\Ret_{\gen{u}}\).
  The evaluation map for \(\Ret_{\gen{u}}\) is given by application of the retraction so in this case by \(\q_{u} \co \gen{u\p_{u}} \to \gen{u}\).
  By construction, \(\q_{u} \cc \p_{u}^+ \sim \id\).
  As the top row is an equivalence, so is the equivalent bottom row.
\end{proof}

\bibliography{references}

\appendix

\section{Locally representable \texorpdfstring{\(\sD\)}{D}-modal types}\label[appendix]{sec:representables}

In this section, we introduce a notion of representability for our notion of higher presheaf.
We then show a version of the Yoneda lemma for these representables.

\subsection{Homotopy representations}

In this entire section, we work internally to \(\CubicalSet\).
A first observation is that for a cubical category \(\cat{C}\), the representable presheaf \(\yo x\) for some \(x \co \cat{C}\) is not necessarily fibrant.
So we still operate under \cref{ass:fibrant-cat}.
From this, we obtain that \(\yo x\) is levelwise fibrant.

Representable presheaves have, by the Yoneda lemma in this internal setting, their expected properties.
A higher presheaf \(\Delta\) should be homotopy representable if it is sufficiently close to a standard representable.
Presenting \(\Delta\) by a global \(\sD\)-modal type, we express this using a \emph{levelwise equivalence from a standard representable}.
A key observation (\cref{prop:cobar-modal-types-see-levelwise-equivalences-between-levelwise-fibrant-types-as-equivalences}) is that the \(\sD\)-modal types see levelwise equivalences between levelwise fibrant types as equivalences.
From this, we conclude a version of the Yoneda lemma for our notion of homotopy representable (cf.~\cref{prop:homotopical-yoneda}).

\begin{definition}
  Let \(\Delta \co \PSh*{\cat{C}}\) be a levelwise fibrant presheaf.
  A \emph{homotopy representation} of \(\Delta\) is a levelwise equivalence from a standard representable presheaf \(\delta \co \yo x \to \Delta\).
\end{definition}

\begin{remark}
  If \(\Delta\) is levelwise fibrant, then all notions of levelwise equivalence coincide.
  In particular the definition via contractible fibers unfolds to the following: a homotopy representation of \(\Delta\) is some \(\delta \co \Delta(x)\) such that \(\sum_{f : \cat{C}(x', x)} \delta f \simeq_{\Delta_{x'}} \delta'\) is contractible for all \(x' \co \cat{C}_0\) and \(\delta' \co \Delta(y)\).
  This can be read as a homotopical way of stating that \(\elems[\cat{C}]\Delta\) has a terminal object.
\end{remark}

\begin{example}
  The unit \(\tau \co \id \to \LO{D}\) is a levelwise equivalence by \cref{prop:cobar-unit-is-levelwise-equivalence}.
  Thus, the \(\sD\)-modal replacement of \(\yo x\) is a homotopy representable presheaf that is fibrant and \(\sD\)-modal.
\end{example}

In order to show our key observation, we record the following standard closure property of equivalences.
Note that we require homogenous fibrancy since it is sufficient to compose homotopies.

\begin{lemma}[Internal to \(\PSh*{\cat{C}}\)]\label[lemma]{prop:equivalences-between-homogenously-fibrant-types-are-closed-under-retract}
  Consider a diagram
  \begin{center}\(\begin{tikzcd}
    A_0
    \ar[r, "s_0"]
    \ar[d, "a"]
  &
    B_0
    \ar[r, "r_0"]
    \ar[d, "b"]
  &
    C_0
    \ar[d, "c"]
  \\
    A_1
    \ar[r, "s_1"]
  &
    B_1
    \ar[r, "r_1"]
  &
    C_1
  \end{tikzcd}\)\end{center}
  of homogeneously fibrant types.
  Assume that the horizontal composites and the middle vertical maps are equivalences.
  Then so are the left and right vertical maps.
\end{lemma}
\begin{proof}
  Define \(k \coloneqq r_0 \cc b^{-1} \cc s_1 \co A_1 \to C_0\) using the given homotopy inverse of \(b\).
  We have
  \begin{align*}
    k \cc a = r_0 \cc b^{-1} \cc s_1 \cc a \sim r_0 \cc b^{-1} \cc b \cc s_0 &\sim r_0 \circ s_0 \rlap{,}
  \\
    c \cc k = c \cc r_0 \cc b^{-1} \cc s_1 \sim r_1 \cc b \cc b^{-1} \cc s_1 &\sim r_1 \cc s_1 \rlap{.}
  \end{align*}
  This makes \(k \cc a\) and \(c \cc k\) into equivalences.
  So \(a\) and \(c\) are equivalences by \(2\)-out-of-\(6\).
\end{proof}

Using the above, we can now show that even between only levelwise fibrant types, the \(\sD\)-modal types see levelwise equivalences as equivalences.
Note that \(\Pi\)-types are homogeneously fibrant if their codomain is.
Hence, in the following statement, all notions of equivalence again coincide.

\begin{lemma}[Internal to \(\PSh*{\cat{C}}\)]\label[lemma]{prop:cobar-modal-types-see-levelwise-equivalences-between-levelwise-fibrant-types-as-equivalences}
  Consider a levelwise equivalence \(f \co A \to B\) between levelwise fibrant types.
  Given a fibrant \(\sD\)-modal family \(Q\) over \(B\), the following map given by restriction is an equivalence:
  \[
    f^* \co \prod_{b\co B} Q(b) \longrightarrow \prod_{a\co A} Q\paren[\big]{f(a)},
    \qquad g \longmapsto gf.
  \]
\end{lemma}
\begin{proof}
  Consider the following diagram of homogeneously fibrant types:
  \begin{center}\(\begin{tikzcd}[column sep=.5cm]
    \displaystyle \smashoperator{\prod_{b \co B}} Q(b)
    \ar[r, "\sD"]
    \ar[d, "f^*", shorten >= -0.2em, shorten <= -0.2em]
  &
    \displaystyle \smashoperator{\prod_{y \co \sD B}} (\widetilde{\sD} Q)(y)
    \ar[r, "\tau_B^*"]
    \ar[d, "(\sD f)^*", shorten >= -0.2em, shorten <= -0.8em]
  &
    \displaystyle \smashoperator{\prod_{b \co B}} (\widetilde{\sD} Q)\paren[\big]{\tau_B(b)}
    \ar[d, "f^*", shorten >= -0.2em, shorten <= -0.8em]
  \\
    \displaystyle \smashoperator{\prod_{a \co A}} Q\paren[\big]{f(a)}
    \ar[r, "\sD"]
  &
    \displaystyle \smashoperator{\prod_{x \co \sD A}} (\widetilde{\sD} Q)\paren[\big]{(\sD f)(x)}
    \ar[r, "\tau_A^*"]
  &
    \displaystyle \smashoperator{\prod_{a \co A}} (\widetilde{\sD} Q)\paren[\big]{\tau_B(f(a))} \rlap{.}
  \end{tikzcd}\)\end{center}
  The horizontal composites are isomorphic to the functorial action of dependent product on the family of maps \(\tau_{Q(b)} \co Q(b) \to \sD(Q(b))\) for \(b \co B\) (and its restriction to \(A\)).
  The latter are equivalences as \(Q\) is \(\sD\)-modal.
  So the horizontal composites are equivalences since the functorial action of dependent product on the family of maps respects homotopies.
  The map \(\sD f\) is an equivalence between fibrant types by \cref{prop:cobar-pseudomorphism:preserves-fibrancy,prop:cobar-strictifies-levelwise-equivalences}.
  Precomposition with an equivalence between fibrant types induces an equivalence of \(\Pi\)-types with a fibrant codomain (see, \eg,~\cite[Theorem 13.4.1]{Rijke2025}).
  This makes the middle vertical map into an equivalence.
  Thus, \(f^*\) is an equivalence by \cref{prop:equivalences-between-homogenously-fibrant-types-are-closed-under-retract}.
\end{proof}

\begin{remark}
  In the special case of \(\tau_A \co A \to \sD A\), the above strengthens of the usual universal property of \(\sD\), weakening the fibrancy assumption on \(A\) to a levelwise one.
\end{remark}

We can now prove an incarnation of the Yoneda lemma in our setting.
The phrasing we choose is akin to what is sometimes called the ``type-theoretical Yoneda lemma'': elements of a type in a representable context are local sections at the representing object and identity.
This is a straightforward corollary of the ordinary Yoneda lemma, using that elements are in bijection with sections of display maps.

\begin{proposition}\label[proposition]{prop:homotopical-yoneda}
  Let \(\delta : \yo x \to \Delta\) be a homotopy representation of a levelwise fibrant presheaf \(\Delta\).
  For every \(P \co \Ty*_{\PSh*{\cat{C}}}^{\sD}(\Delta)\), the following map is an equivalence
  \[
    \El*_{\PSh*{\cat{C}}}(\Delta, P) \longrightarrow P(x, \delta), \qquad
    u \longmapsto u(x, \delta).
  \]
\end{proposition}
\begin{proof}
  By \cref{prop:cobar-modal-types-see-levelwise-equivalences-between-levelwise-fibrant-types-as-equivalences}, we have the equivalence \(\prod_{f \co \yo x} P(\delta f) \to \prod_{\delta' \co \Delta} P\delta'\) in \(\PSh*{\cat{C}}\).
  Since \(\PSh*{\cat{C}}(1, -)\) preserves homotopies, we obtain an equivalence of types of elements, and thus
  \[
    \El*_{\PSh*{\cat{C}}}(\Delta, P)
    \cong \El*_{\PSh*{\cat{C}}}(1, \Pi_{\Delta} P)
    \simeq \El*_{\PSh*{\cat{C}}}(1, \Pi_{\yo x} P\delta)
    \cong \El*_{\PSh*{\cat{C}}}(\yo x, P\delta)
    \cong P(x, \delta).
  \]
  Unfolding shows that the underlying map of the above equivalence is the desired map.
\end{proof}

\subsection{Locally homotopy representable types}

In this entire section, we work internally to \(\CubicalSet\).
With a notion of representable presheaf in our setting, we can also introduce the corresponding relative notion, which we call a \emph{locally homotopy representable type}.
This is analogous to the usual notion of a (locally) representable natural transformation.
This notion is originally due to Grothendieck and is used by Awodey~\cite{Awodey2018} to define natural models of type theory.

\begin{definition}\label[definition]{def:representables:locally-representable}
  Let \(A \co \Ty*_{\PSh*{\cat{C}}}^{\Fib*_0}(\Gamma)\) be a type.
  A \emph{local homotopy representation} of \(A\) is a homotopy representation of \(\yo x.A\gamma \co \PSh*{\cat{C}}\) for all \(x \co \cat{C}_0\) and \(\gamma \co \Gamma_x\).
  A \emph{locally homotopy representable type} is a type \(A\) with a local homotopy representation.
  For \(\gamma \co \Gamma_x\) we write \(x.\gamma \co \cat{C}_0\) with \(\p_\gamma \co \cat{C}(x.\gamma, x)\) and \(\q_\gamma \co A(x.\gamma, \gamma\p_\gamma)\) for the elements determining the homotopy representation.
\end{definition}

\begin{remark}\label{rem:locally-representable-is-structure}
  Note that the above is strictly stable under substitution.
  If \(\sigma \co \Delta \to \Gamma\) and \(A \co \Ty*(\Gamma)\) has a local homotopy representation, then so does \(A\sigma \co \Ty*(\Delta)\) given for \(\delta \co \Delta_x\) by \(x.\sigma\delta\), \(\p_{\sigma\delta}\), and \(\q_{\sigma\delta} \co A(x.\sigma\delta, \sigma\delta\p_{\sigma\delta})\).
\end{remark}

A well-known fact in the \(1\)-categorical setting is that \((A^{\yo x})(y) \cong A(y \times x)\), provided enough products exist in the base category.
More generally, \((\Pi_{\yo x}A)(y) \cong A(y \times x, \pi_1)\).
For locally representable natural transformations a generalization of this holds true.
If \(A\) is a type in context \(\Gamma\) whose display map is locally representable, one has, in the above notation, that \((\Pi_AB)(x, \gamma) \cong B(x.\gamma, \gamma\p, \q)\).
The latter is a generalization since \(\yo x \to 1\) is locally representable exactly if the base category has all products with \(x\).
We show an analogue of this characterization of \(\Pi\)-types with locally homotopy representable domain in our setting.

\begin{lemma}\label[lemma]{prop:pushforward-along-homotopy-locally-representable-transformation}
  Naturally in \(\Gamma\), given a locally homotopy representable type \(A \co \Ty*^{\Fib*_0}(\Gamma)\) and a \(\sD\)-modal family \(P \co \Ty*^{\sD}(\Gamma.A)\), the following map is an equivalence for \(x \co \cat{C}\) and \(\gamma \co \Gamma_x\):
  \begin{equation*}
    (\Pi_AP)(x, \gamma) \longrightarrow P(x.\gamma, \gamma\p_\gamma, \q_{\gamma}),
    \qquad
    u \longmapsto u_{\p_{\gamma}}(\q_\gamma).
  \end{equation*}
\end{lemma}
\begin{proof}
  By Yoneda, the universal property of \(\Pi\), and \cref{prop:homotopical-yoneda}, we have
  \begin{multline*}
    (\Pi_AP)(x, \gamma)
    \cong \El*\paren[\big]{\yo x, \paren{\Pi_AP}\gamma }
    \cong \El*\paren{\yo x.A \gamma, P\gamma^+}
    \simeq P\gamma^+(x.\gamma, \q_\gamma)
    = P(x.\gamma, \gamma\p_\gamma, \q_\gamma).
  \end{multline*}
  This is the desired map since the equivalence \(\El*\paren{\yo x.A \gamma, P\gamma^+} \to P\gamma^+(x.\gamma, \p_\gamma, \q_\gamma)\) from \cref{prop:homotopical-yoneda} is given by evaluation at \(x.\gamma\) and \((\p_\gamma, \q_\gamma) \co \paren{\yo x.A\gamma}(x.\gamma)\).
\end{proof}

\subsection{Projectivity of locally representable types}

When working in homotopy type theory, one usually defines projectivity of \(A\) to mean that the unique map \(\varnothing \to A\) lifts against surjections.
Crucially, this definition contains a quantification over \emph{all} types, so even in the presence of universes we cannot internalize this definition as a single type.
Instead, such definitions are usually phrased externally, quantifying over all substitutions into the current context and all types in these later contexts.

\begin{definition}\label{def:projective}
  A type \(A \in \Ty(\Gamma)\) is \emph{projective} if, naturally in \(\sigma \co \Delta \to \Gamma\), for all \(P \in \Ty(\Delta.A\sigma)\), there is a map \(\El(\Delta, \Pi_{A\sigma}\trunc{P}) \to \El(\Delta, \trunc{\Pi_{A\sigma}P})\).
\end{definition}

\begin{remark}
  Similar to local homotopy representability (\cf~\cref{rem:locally-representable-is-structure}), this is a substitution-stable notion.
  Given a projective type \(A \in \Ty(\Gamma)\) and a substitution \(\sigma \co \Delta \to \Gamma\), then the type \(A\sigma \in \Ty(\Delta)\) is also projective.
  Given \(\tau \co \Theta \to \Delta\) and \(P \in \Ty(\Theta, P)\) we define the map \(\El(\Theta, \Pi_{A\sigma\tau}\trunc{P}) \to \El(\Theta, \trunc{\Pi_{A\sigma\tau} P})\) as the one given by \(A\) for \(\tau\sigma \co \Theta \to \Gamma\) and \(P\).
\end{remark}

\begin{remark}
  In the \(1\)-categorical setting, there are different notions of projectivity.
  Usually, an object \(A\) is called \emph{projective} if it lifts against epimorphisms, while it is called \emph{internally projective} if \((-)^A\) preserves epimorphisms.
  The latter aligns with the definition in models of type theory when considering a closed type.
\end{remark}

In a presheaf topos, representables are projective, but not necessarily internally projective.
They are so in the presence of binary products in the base category.
This is exactly the situation in which \(\yo x \to 1\) is locally representable.
The next lemma can be seen as a version of this fact in our setting, generalized to types that are not necessarily closed.

\begin{proposition}\label[proposition]{prop:lifting-of-lex-operations-over-fibrant-sites:locally-homotopy-representable-types-are-projective}
  Naturally in \(\Gamma\), every locally homotopy representable type \(A \in \Ty^{\sD\LexOp}_{\PSh{\cat*{C}}}(\Gamma)\) is projective.
\end{proposition}
\begin{proof}
  Since being locally homotopy representable is stable under substitution, it suffices to consider \(P \in \Ty^{\sD\LexOp}(\Gamma.A)\).
  Internally to \(\CubicalSet\), we have a map \(\El*_{\PSh*{\cat{C}_0}}(\PM{U}\Gamma.\PM{U}A, \trunc{UP}) \to \El*_{\PSh*{\cat{C}_0}}(\Gamma, \trunc{\PM{U}\paren{\Pi_AP}})\).
  It is given by restriction followed by the inverse from \cref{prop:pushforward-along-homotopy-locally-representable-transformation}:
  \begin{equation*}
    {\prod_{\substack{ x \co \cat{C}_0 \\ (\gamma, a) \co (\Gamma.A)_x }}} \trunc{P(x, \gamma, a)} 
    \longrightarrow {\prod_{\substack{ x \co \cat{C}_0 \\ \gamma \co \Gamma_x }}} \trunc{P(x.\gamma, \gamma\p_\gamma, \q_\gamma)} 
    \overset{\simeq}\longleftarrow {\prod_{\substack{ x \co \cat{C}_0 \\ \gamma \co \Gamma_x }}} \trunc{ \paren{\Pi_AP}(x, \gamma) }. 
  \end{equation*}
  Hence, by externalization of the above map we obtain
  \begin{align*}
    \El_{\PSh{\cat*{C}}}(\Gamma, \Pi_A\trunc{P})
    &\longrightarrow \El_{\PSh{\cat*{C}}}(\Gamma.A, \trunc{P}) \tag{universal property of \(\Pi\)} \\
    &\longrightarrow \El_{\PSh{\cat*{C}_0}}(\PM{U}\Gamma.\PM{U}A, \trunc{\PM{U}P}) \tag{\cref{prop:lifting-of-lex-operations-over-fibrant-sites:levelwise-principle}} \\
    &\longrightarrow \El_{\PSh{\cat*{C}_0}}(\PM{U}\Gamma, \trunc{\PM{U}\paren{\Pi_AP}}) \\
    &\longrightarrow \El_{\PSh{\cat*{C}}}(\Gamma, \trunc{\Pi_AP}). \tag{\cref{prop:lifting-of-lex-operations-over-fibrant-sites:levelwise-principle}}
  \end{align*}
  This yields the desired logical implication.
\end{proof}

\end{document}